\documentclass{ecai} 

\usepackage{latexsym}
\usepackage{amssymb}
\usepackage{amsmath}
\usepackage{amsthm}
\usepackage{booktabs}
\usepackage{enumitem}
\usepackage{graphicx}
\usepackage{color}

\usepackage{amsfonts}
\usepackage{mathtools}
\usepackage{pifont,dsfont} 
\usepackage{tabularx,lipsum,environ}
\usepackage{multirow}
\usepackage{array}
\usepackage{makecell}

\newtheorem{theorem}{Theorem}
\newtheorem{lemma}[theorem]{Lemma}

\newtheorem{definition}{Definition}
\newtheorem{example}{Example}

\newcommand{\BibTeX}{B\kern-.05em{\sc i\kern-.025em b}\kern-.08em\TeX}

\DeclarePairedDelimiter\ceil{\lceil}{\rceil}
\DeclarePairedDelimiter\floor{\lfloor}{\rfloor}

\newcommand{\littlep}{{\rm p}}
\newcommand{\manyone}{\ensuremath{\mbox{$\,\leq_{\rm m}^{{\littlep}}$\,}}}

\newcommand{\EP}[3]{
	\begin{center}
		\smallskip
		{\small 
			\begin{tabularx}{\columnwidth}{@{}l@{\hspace*{2mm}}l@{}}
				\toprule
				\multicolumn{2}{c}{\sc{#1}} \\
				\midrule
				{\bf Given:}& \parbox[t]{0.80\columnwidth}{#2\vspace*{1mm}} \\
				{\bf Question:}& \parbox[t]{0.80\columnwidth}{#3\vspace*{.5mm}} \\ 
				\bottomrule
			\end{tabularx}
		}
		\smallskip
	\end{center}
}

\newcommand{\proofonlyif}{\smallskip\textit{Only if:\quad}}
\newcommand{\proofif}{\smallskip\textit{If:\quad}}

\newcommand{\PenroseBanzhaf}{\beta}

\newcommand{\p}{\ensuremath{\mathrm{P}}}
\newcommand{\np}{\ensuremath{\mathrm{NP}}}
\newcommand{\conp}{\ensuremath{\mathrm{coNP}}}

\newcommand{\PP}{\ensuremath{\mathrm{PP}}}

\newcommand{\PH}{\ensuremath{\mathrm{PH}}}

\newcommand{\OMIT}[1]{}

\newenvironment{proofs}{\noindent{\textsc{Proof.}}}{\literalqed\bigskip}
\def\literalqed{{\ \nolinebreak\hfill\mbox{$\square$\quad}}}

\newenvironment{proofsketch}{{\textsc{Proof Sketch.}}}

\begin{document}

\begin{frontmatter}

\paperid{3527}

\title{Control by Deleting Players from Weighted Voting Games Is NP$^{\text{PP}}$-Complete for the Penrose--Banzhaf Power Index}

\author[A]{\fnms{Joanna}~\snm{Kaczmarek}\orcid{0000-0001-6652-6433}\thanks{Corresponding Author. Email: joanna.kaczmarek@hhu.de}%
}
\author[A]{\fnms{J\"{o}rg}~\snm{Rothe}\orcid{0000-0002-0589-3616}%
}

\address[A]{Institut f\"{u}r Informatik, MNF, Heinrich-Heine-Universit\"{a}t D\"{u}sseldorf, D\"{u}sseldorf, Germany}

\begin{abstract}
Weighted voting games are a popular class of coalitional games that are widely used to model real-life situations of decision-making.
They can be applied, for instance, to analyze legislative processes in parliaments or voting in corporate structures.
Various ways of tampering with these games have been studied, among them merging or splitting players, 
fiddling with the quota, and controlling weighted voting games by adding or deleting players.
While the complexity of control by adding players to such games so as to change or maintain a given player's power has been recently settled, the complexity of control by deleting players from such games (with the same goals) remained open.
We show that when the players' power is measured by the probabilistic Penrose--Banzhaf index, some of these problems are complete for $\np^{\PP}$---the class
of problems solvable by $\np$ machines equipped 
with a $\PP$ (``probabilistic polynomial time'') oracle.
Our results optimally improve the currently known lower bounds of
hardness for much smaller complexity classes, thus providing protection
against SAT-solving techniques in practical applications.
\end{abstract}

\end{frontmatter}

\section{Introduction}

Due to their numerous applications in real-life decision-making, weighted voting games (WVGs), a popular class of simple coalitional games, have been studied intensively for the last decades (see the book chapters by Bullinger \emph{et al.}~\cite{bul-elk-rot:b-2nd-edition:economics-and-computation-cooperative-game-theory} and Chalkiadakis and Wooldridge~\cite{cha-woo:b:handbook-comsoc-weighted-voting-games}). 

WVGs have been applied to model and analyze collective decision-making in legislative processes and for voting in parliaments~\cite{sha:b:polsci:power} (e.g., the parliament of the European Union) or in corporate structures (such as the International Monetary Fund~\cite{gam:j:power-indices-for-political-and-financial-decision-making-review}); it can also be used for making decisions in joint stock companies where the shareholders' stock values specify their weights, etc.

A central question with regard to WVGs is how significant the players in them are, that is, how much influence they have when it comes to forming winning coalitions.
This is typically measured by so-called power indices.
We focus on the \emph{probabilistic Penrose--Banzhaf index}, which is due to Dubey and Shapley~\cite{dub-sha:j:banzhaf} who were seeking to avoid certain disadvantages of the \emph{normalized Penrose--Banzhaf index}, originally proposed by Penrose~\cite{pen:j:banzhaf-index} and later re-invented by Banzhaf~\cite{ban:j:weighted-voting-doesnt-work}.
Prasad and Kelly~\cite{pra-kel:j:voting} showed that computing the probabilistic Penrose--Banzhaf index is $\#\p$-complete even under \emph{parsimonious} reductions (which preserve the number of solutions).
Here, $\#\p$ denotes Valiant's ``counting class'' of functions that give the number of solutions of $\np$ problems~\cite{val:j:permanent}.

Certain ways of how to tamper with the outcome of WVGs have been studied.
For example, Zuckerman \emph{et al.}~\cite{zuc-fal-bac-elk:j:manipulating-quote-in-wvg} considered the consequences of manipulating the quota of a given WVG.
Aziz \emph{et al.}~\cite{azi-bac-elk-pat:j:false-name-manipulations-wvg} introduced the notions of merging, splitting,\footnote{``Splitting players'' is also known as ``false-name manipulation.''}
and annexing players in WVGs and studied the complexity of whether such actions can be beneficial in terms of raising a distinguished player's power for 
the Shapley--Shubik~\cite{sha-shu:j:shapley-shubik-index} 
and the normalized Penrose--Banzhaf index.
Later on, Rey and Rothe~\cite{rey-rot:j:false-name-manipulation-PP-hard} improved their $\np$-hardness lower bounds to $\PP$-hardness for merging and splitting players with respect to the Shapley--Shubik index (and established $\PP$-hardness of these problems also for the probabilistic Penrose--Banzhaf index), where $\PP$ denotes the complexity class ``probabilistic polynomial time,'' which is due to Gill~\cite{gil:j:probabilistic-tms} and contains both $\np$ and $\conp$, 
and even $\p^{\np[\mathcal{O}(\log n)]}$~\cite{bei-hem-wec:j:powerprob},
a.k.a.\ $\Theta_2^{p}$,
the class of problems solvable in deterministic polynomial time with at most logarithmically many queries to an $\np$ oracle.
Faliszewski and Hemaspaandra~\cite{fal-hem:j:power-index-comparison} proved that comparing a given player's Shapley--Shubik or probabilistic Penrose--Banzhaf index in two given WVGs is $\PP$-complete.

Rey and Rothe~\cite{rey-rot:j:structural-control-in-weighted-voting-games} introduced control by either adding players to or removing them from a given WVG, and they established the first complexity results for the corresponding problems.
Their work was continued by Kaczmarek and Rothe~\cite{kac-rot:j:controlling-weighted-voting-games-by-deleting-or-adding-players-with-or-without-changing-the-quota,kac-rot:c:control-by-adding-players-to-change-or-maintain-shapley-shubik-or-the-penrose-banzhaf-power-index-in-WVGs-is-complete-for-np-pp,kac-rot:t:control-by-adding-players-to-change-or-maintain-shapley-shubik-or-the-penrose-banzhaf-power-index-in-WVGs-is-complete-for-np-pp} who, in particular, improved some of their complexity results (and also introduced another type of these problems where the game's quota is changed proportionally with the number of players that is changed by adding or deleting them).
In particular, control by adding players to a given WVG with the goal of either changing or maintaining a distinguished player's Shapley--Shubik or probabilistic Penrose--Banzhaf index was recently shown to be even $\np^{\PP}$-complete~\cite{kac-rot:c:control-by-adding-players-to-change-or-maintain-shapley-shubik-or-the-penrose-banzhaf-power-index-in-WVGs-is-complete-for-np-pp,kac-rot:t:control-by-adding-players-to-change-or-maintain-shapley-shubik-or-the-penrose-banzhaf-power-index-in-WVGs-is-complete-for-np-pp}, where $\np^\PP$ is the class of problems solvable by a nondeterministic Turing machine having access to a $\PP$ oracle.
$\np^\PP$ is an extraordinarily large complexity class; for example, the entire polynomial hierarchy is contained in it by Toda's celebrated result~\cite{tod:j:pp-ph}.

So, while the complexity of control by \emph{adding} players to WVGs with the goal of increasing, nonincreasing, decreasing, nondecreasing, or maintaining a given player's power has been settled already for both the Shapley--Shubik and the probabilistic Penrose--Banzhaf index, the complexity of control by \emph{deleting} players from such games (with the same goals) has remained open by now.
Building on the proof technique of Kaczmarek and Rothe~\cite{kac-rot:c:control-by-adding-players-to-change-or-maintain-shapley-shubik-or-the-penrose-banzhaf-power-index-in-WVGs-is-complete-for-np-pp,kac-rot:t:control-by-adding-players-to-change-or-maintain-shapley-shubik-or-the-penrose-banzhaf-power-index-in-WVGs-is-complete-for-np-pp}, we establish optimal (i.e., $\np^{\PP}$-completeness) results for some of these problems (specifically, for the goals of decreasing, nonincreasing, and maintaining a given player's probabilistic Penrose--Banzhaf index).
Note, however, that the constructions needed for the case of deleting players are way more involved than in the case of adding players.

In order to motivate these problems, recall that WVGs are the default framework to model decision-making in legislative bodies such as the EU parliament,\footnote{A bit more precisely speaking, decision-making in the EU parliament can be described as a three-dimensional vector weighted voting game (see, e.g., the book by Chalkiadakis \emph{et al.}~\cite{cha-elk-woo:b:computational-aspects-of-cooperative-game-theory} and the book chapter by Bullinger \emph{et al.}~\cite{bul-elk-rot:b-2nd-edition:economics-and-computation-cooperative-game-theory}).}
with population sizes roughly specifying the countries' weights.
Specifically, to motivate control by adding players, Kaczmarek and Rothe~\cite[p.~1]{kac-rot:c:control-by-adding-players-to-change-or-maintain-shapley-shubik-or-the-penrose-banzhaf-power-index-in-WVGs-is-complete-for-np-pp} write,
\emph{`The EU is
	constantly expanding: New members join in (or, rarely, they leave),
	which is exactly
	control by adding players,
	raising the question of if and how the power of old EU members is
	changed by adding new ones to the EU---just one clear-cut case of
	motivation among various others.  If new members join, an old one may
	insist on having the same power afterwards (motivating the goal of
	``maintaining one's power''), or at least not lose power
	(``nondecreasing one's power''), or Poland may insist that Germany's
	power does not increase when Ukraine joins (``nonincreasing one's
	power'').'}
Similarly, control by deleting players is motivated by, for instance, the Brexit: the United Kingdom leaving the EU in 2020.
In such a case, it is natural that the remaining countries may request to maintain their power afterwards, or that, say, Luxembourg may insist that its power at least does not decrease due to the United Kingdom leaving.

Another possible motivation could be the desire to make someone less powerful, without having any particular preference about who specifically (if anyone at all) might become more powerful instead.
For instance, in our example, there might have been countries that were more closely connected with the United Kingdom and, as a result of the Brexit, suffered greater losses than other EU members.

In addition to the motivation of our problems by such real-world scenarios, studying control by deleting players from (or adding them to) a given WVG is also interesting from a purely scholarly point of view. 
Specifically, control attacks on a plenty of voting rules (such as control by adding or deleting either voters or candidates) have been thoroughly investigated 
for decades~\cite{bau-rot:b-2nd-edition:economics-and-computation-preference-aggregation-by-voting,fal-rot:b:handbook-comsoc-control-and-bribery}.
Control attacks in cooperative game theory (e.g., by deleting players from or adding them to WVGs), however, have been pretty much neglected so far (except for the papers mentioned above), which may seem a bit surprising.

We continue the work of Rey and Rothe~\cite{rey-rot:j:structural-control-in-weighted-voting-games} and Kaczmarek and Rothe~\cite{kac-rot:j:controlling-weighted-voting-games-by-deleting-or-adding-players-with-or-without-changing-the-quota} (applying their techniques from~\cite{kac-rot:c:control-by-adding-players-to-change-or-maintain-shapley-shubik-or-the-penrose-banzhaf-power-index-in-WVGs-is-complete-for-np-pp,kac-rot:t:control-by-adding-players-to-change-or-maintain-shapley-shubik-or-the-penrose-banzhaf-power-index-in-WVGs-is-complete-for-np-pp}) on the computational complexity of structural control by deleting players from a given WVG.
Specifically, for the Penrose--Banzhaf power index, Rey and Rothe~\cite{rey-rot:j:structural-control-in-weighted-voting-games} showed $\conp$-hardness for the goals of maintaining, nonincreasing, or decreasing a distinguished player~$p$'s power, and Kaczmarek and Rothe~\cite{kac-rot:j:controlling-weighted-voting-games-by-deleting-or-adding-players-with-or-without-changing-the-quota} slightly improved this lower bound for the latter goal of decreasing $p$'s power to 
$\Theta_2^{p}$-hardness.

We improve these previous $\conp$-hardness and 
$\Theta_2^{p}$-hardness
lower bounds to $\np^{\PP}$-completeness, a much larger complexity class.
Thus, for these three problems, we now have matching upper and lower bounds.
Due to modern SAT-solvers and ILP techniques, $\np$-hardness (or $\conp$- or even 
$\Theta_2^{p}$-hardness)
results do not provide real protection against control attacks in practice.

However, and this is very significant for practical applications,
tackling $\np^{\PP}$-hard problems is completely out of reach for these techniques, as $\np^{\PP}$ is so huge.
For example, Baumeister \emph{et al.}~\cite{bau-jae-neu-nis-rot:j:acceptance-in-incomplete-argumentation-frameworks} were able to provide stronger refinements for SAT-based CEGAR algorithms for deciding their problems from argumentation theory, evidenced in practice by noticeably improved runtimes.
However, their problems are complete for ``only'' $\Sigma_3^p = \mathrm{NP}^{\mathrm{NP}^{\mathrm{NP}}}$, and we are not aware of any successful attempts of applying SAT-solving techniques to higher levels of the polynomial hierarchy than the third level, not to mention to a class as large as~$\np^{\PP}$, which contains the entire hierarchy
$\PH$, the class $\PP^{\PH}$, and even $\np^{\PP^{\PH}}$~\cite{tod:j:pp-ph,hem:c:comsoc-and-complexity-theory-bffs}.
We would like to point out a technical difference between the problems of control by adding players and those of control by deleting players.
In the former case, two sets of players are given: players who are already in the given game and those who can be added to it.
In the latter case, all players defined through the reduction are already part of the game.
This introduces a new challenge when applying the technique by Kaczmarek and Rothe~\cite{kac-rot:c:control-by-adding-players-to-change-or-maintain-shapley-shubik-or-the-penrose-banzhaf-power-index-in-WVGs-is-complete-for-np-pp,kac-rot:t:control-by-adding-players-to-change-or-maintain-shapley-shubik-or-the-penrose-banzhaf-power-index-in-WVGs-is-complete-for-np-pp}, namely the necessity to know the \emph{exact} number of solutions for an instance of the problem  being reduced.
In this paper, we show how to handle this situation, which technically requires substantially more effort than in the case of control by adding players. 

We provide the needed notions from cooperative game theory and complexity theory in Section~\ref{sec:preliminaries}.
In Section~\ref{sec:control-problems}, we present our results.
We conclude in Section~\ref{sec:conclusions} and outline some open questions.
Due to space limitations, some of our proofs are omitted here and are moved to the appendix.

\section{Preliminaries}
\label{sec:preliminaries}

We start by recalling some notions from cooperative game theory, especially those needed to define weighted voting games, which are the subject of our paper.

\begin{definition}
  Let $N = \{1, \ldots , n\}$ be a \emph{set of players}, and let
  \[
  v~:~2^N \rightarrow \mathbb{R}_{+} \cup \{0\}
  \]
  be a \emph{characteristic function} that maps each subset of $N$ (called a \emph{coalition}) to some nonnegative real value.
	The pair $(N,v)$ is called a \emph{coalitional game}, and it is said to be \emph{simple} if it is \emph{monotonic} (i.e., for all $T,T'$ with $T \subseteq T' \subseteq N$, $v(T) \leq v(T')$) and $v(T) \in \{0,1\}$ for each coalition $T \subseteq N$.
\end{definition}

Weighted voting games as a special type of simple coalitional games are defined as follows.
\begin{definition}
	A \emph{weighted voting game (WVG)}
	$\mathcal{G} = (w_{1}, \ldots, w_{n};q)$
	is a simple coalitional game with player set $N$ and characteristic function $v$ defined as follows: $v(T) = 1$ if $w_{T} \ge q$, and $v(T) = 0$ otherwise, where $q$ is a natural number
	called the \emph{quota} and $w_{T}~=~\sum_{i \in T} w_{i}$ with $w_{i}$ being the nonnegative integer \emph{weight of player~$i \in N$}. 

        Coalition~$T$ is said to be a \emph{winning coalition} if $v(T)=1$,
	and it is called a \emph{losing coalition} if $v(T)=0$.

        Moreover, we call a player~$i$ \emph{pivotal for coalition
	  $T\subseteq N\setminus\{i\}$} if
        \[
        v(T\cup \{i\})-v(T)=1.
        \]
\end{definition}

One of the things we want to know about players is how significant they are in a given game.
We usually measure this by so-called \emph{power indices}.
The main information used in determining the power index of a player~$i$ is the number of coalitions $i$ is pivotal for, i.e., the number of coalitions that would lose unless $i$ joins them.
We study one of the most popular and well-known power indices: the \emph{probabilistic Penrose--Banzhaf power index}.
This index was introduced by Dubey and Shapley~\cite{dub-sha:j:banzhaf} as an
alternative to the original \emph{normalized Penrose--Banzhaf index}~\cite{pen:j:banzhaf-index,ban:j:weighted-voting-doesnt-work}.
It is defined as follows.

\begin{definition}\label{PBI}
	Let $\mathcal{G}$ be a WVG.
	The \emph{probabilistic Penrose--Banzhaf power index of a player~$i$ in
		$\mathcal{G}$} is defined by
	\[
	\PenroseBanzhaf(\mathcal{G},i) =
	\frac{1}{2^{|N|-1}}\sum\limits_{T \subseteq N \setminus \{i\}}(v(T \cup \{i\})-v(T)).
	\]
\end{definition}

We assume that the readers are familiar with the fundamental concepts of complexity theory, such as the well-known complexity classes $\p$ and $\np$, as well as $\conp = \{\overline{L} \mid L \in \np\}$ and the class ``probabilistic polynomial time,'' denoted by $\PP$ and defined by Gill~\cite{gil:j:probabilistic-tms} as the class of problems that can be solved by a nondeterministic polynomial-time Turing machine accepting its input if and only if at least half of its computation paths accept.
$\np^{\PP}$ is a class in the counting hierarchy (see, e.g., the work of Wagner~\cite{wag:j:succinct}, Parberry and Schnitger~\cite{par-sch:j:parallel-computation-with-threshold-functions}, Toda~\cite{tod:j:pp-ph}, and Tor\'{a}n~\cite{tor:c:ch,tor:j:quantifiers,tor:thesis:count}).
It is defined as the class of problems that can be solved by an $\np$
oracle Turing machine~\cite{tur:j:systems-of-logic} that has access to an oracle set from~$\PP$.

Our hardness and completeness proofs employ the common notion of polynomial-time many-one reducibility~\cite{kar:b:reducibilities}: We say a problem $A$ \emph{(polynomial-time many-one) reduces to a problem $B$} ($A \manyone B$, for short) if there exists a polynomial-time computable function $\rho$ such that for each input~$x$,
\[
x \in A \Longleftrightarrow \rho(x) \in B.
\]
Further, we say $D$ is \emph{hard for some complexity class $\mathcal{C}$} if for each $C \in \mathcal{C}$, $C \manyone D$, and $D$ is said to be \emph{complete for $\mathcal{C}$} if $D \in \mathcal{C}$ and $D$ is $\mathcal{C}$-hard.

For more background in computational complexity, we refer to the common textbooks by Garey and Johnson~\cite{gar-joh:b:int}, 
Papadimitriou~\cite{pap:b:complexity}, and Rothe~\cite{rot:b:cryptocomplexity}. 

For our proofs, we will use the following two problems:
First, de Campos \emph{et al.}~\cite{cam-sta-wey:c:complexity-stochastic-optimization} introduced the problem

\EP{\textsc{E-Minority-SAT}}
{A boolean formula $\phi$ with $n$ variables $x_1,\ldots,x_n$ and an
integer~$k$, $1\le k\le n$.}
{Is there an assignment to $x_1,\ldots,x_k$, the first $k$ variables,
such that at most half of the assignments to the remaining $n-k$
variables $x_{k+1},\ldots,x_n$ satisfy~$\phi$?}
and proved it $\np^{\PP}$-complete;
second, Kaczmarek and Rothe
\cite{kac-rot:c:control-by-adding-players-to-change-or-maintain-shapley-shubik-or-the-penrose-banzhaf-power-index-in-WVGs-is-complete-for-np-pp,kac-rot:t:control-by-adding-players-to-change-or-maintain-shapley-shubik-or-the-penrose-banzhaf-power-index-in-WVGs-is-complete-for-np-pp} considered the analogous problem
where they ask for an assignment to the first $k$ variables such that \emph{exactly} a given number of assignments to the remaining variables satisfies the given formula, 
\OMIT{
\EP{\textsc{E-Exact-SAT}}
{A boolean formula $\phi$ with $n$ variables $x_1,\ldots,x_n$, an
integer~$k$, $1\le k\le n$, and an integer~$\ell$.}
{Is there an assignment to the first $k$ variables $x_1,\ldots,x_k$
such that \emph{exactly} $\ell$ assignments to the remaining
$n-k$ variables $x_{k+1},\ldots,x_n$ 
satisfy~$\phi$?}
} %
and proved its $\np^{\PP}$-completeness as well---we use its special case with the same complexity.

Without loss of generality, we may assume that no given formula in conjunctive normal form (CNF) contains some variable $x$ and its negation $\neg x$ in any of its clauses (and this property can be checked in polynomial time) because otherwise, the clause would be true for any possible truth assignment and could be omitted.
We will also assume that our inputs for the above problems contain only those variables that actually occur in the given boolean formula.

Rey and Rothe~\cite{rey-rot:j:structural-control-in-weighted-voting-games} defined problems describing control by deleting players from a given WVG so as to change a distinguished player's power in the modified game.
When the goal is to decrease this power in terms of the probabilistic Penrose--Banzhaf index (abbreviated as \textsc{PBI}), this control problem is defined as follows:

\EP{\textsc{Control-by-Deleting-Players-to-Decrease-PBI}}
{A weighted voting game $\mathcal{G}$ with player set $N$,
a distinguished player~$i \in N$, and a positive integer
$k< |N|$.}
{Can at most $k$ players $M \subseteq N\setminus\{i\}$
be removed from $\mathcal{G}$ such that for the new game
$\mathcal{G}_{\setminus M}$ with player set $N\setminus M$, it holds that
$\PenroseBanzhaf(\mathcal{G}_{\setminus M},i) < \PenroseBanzhaf(\mathcal{G},i)$?
}

Analogously, they define the corresponding control problems for the goals of nonincreasing or maintaining the probabilistic Penrose--Banzhaf index of the distinguished player by changing the relation sign in the question to  ``$\le$'' and ``$=$,'' respectively.
We also assume that at least one player must be removed for the goals of nonincreasing or maintaining the distinguished player's power, for otherwise, these control problems would be trivial since we would reach these goals already by making no changes at all to the given game.

\begin{example}
  Let
  \[
  \mathcal{G} = (1,2,2,2,3,3;8)
  \]
  be a weighted voting game.
Let us focus on one of the players with weight~$2$ and let us call this player~$p$.
Obviously, $p$ is pivotal for coalitions of weights~$6$ and $7$, and there are eight such coalitions:
\begin{itemize}
  \item a coalition with both players of weight~$3$,
  \item this coalition extended with the player of weight~$1$,
  \item two coalitions with the other two players of weight~$2$ and the player with weight~$3$, and
  \item four coalitions consisting of three players with weights~$1$, $2$, and~$3$.
\end{itemize}
Hence, $p$'s probabilistic Penrose--Banzhaf index is
\[
\PenroseBanzhaf(\mathcal{G},p) = \frac{8}{2^5} = \frac{1}{4}.
\]
Now, remove one of the players of weight~$3$ from~$\mathcal{G}$, and let us call this player~$d$.
Then $p$ stays pivotal for the coalitions without~$d$, so $p$'s probabilistic Penrose--Banzhaf index decreases to
\[
\PenroseBanzhaf(\mathcal{G}_{\setminus\{d\}},p) = \frac{3}{2^{4}} = \frac{3}{16}.
\]
Next, delete one of the players of weight~$2$ instead of~$d$ from~$\mathcal{G}$---let us call this player~$e$.
Then we have
\[
\PenroseBanzhaf(\mathcal{G}_{\setminus\{e\}},p) = \frac{4}{2^{4}} = \frac{1}{4},
\]
i.e., player~$p$ is as powerful in the new game as in the old game.
\end{example}

Rey and Rothe~\cite{rey-rot:j:structural-control-in-weighted-voting-games} observed that $\np^{\PP}$ is the best known upper bound for the problems of control by deleting players from a given WVG, and for the specific goals of decreasing, nonincreasing, or maintaining a distinguished player~$p$'s Penrose--Banzhaf power index, they established $\conp$-hardness as a first lower bound.
There is a glaring gap in complexity between this upper and this lower bound.
Kaczmarek and Rothe~\cite{kac-rot:j:controlling-weighted-voting-games-by-deleting-or-adding-players-with-or-without-changing-the-quota}
took a first step to close this gap by showing 
$\Theta_2^{p}$-hardness
for the goal of decreasing $p$'s power, but still, the remaining gap in complexity can still be called huge (recall that 
$\Theta_2^{p}$
lies in the second level of the polynomial hierarchy, which by Toda's result~\cite{tod:j:pp-ph} is entirely contained in $\p^{\PP} \subseteq \np^{\PP}$).
Our goal is to raise these previous lower bounds to $\np^{\PP}$-hardness, thus establishing $\np^{\PP}$-completeness for these three problems, which is the best result possible as the upper and lower bounds coincide.

To this end, we build on the novel proof technique of Kaczmarek and Rothe~\cite{kac-rot:c:control-by-adding-players-to-change-or-maintain-shapley-shubik-or-the-penrose-banzhaf-power-index-in-WVGs-is-complete-for-np-pp,kac-rot:t:control-by-adding-players-to-change-or-maintain-shapley-shubik-or-the-penrose-banzhaf-power-index-in-WVGs-is-complete-for-np-pp} who recently showed $\np^{\PP}$-hardness of control problems by \emph{adding} players to a given WVG.
In particular, they are using the following transformation from boolean formulas to weight vectors.

\begin{definition}
\label{def:prereduction}
Let $\phi$ be a given boolean formula in CNF
with variables $x_1,\dots,x_n$ and $m$ clauses.
Let $k\in\mathbb{N}$ with $k\le n$ and $r = \ceil{\log_{2} n}-1$.
Define the weight vectors (which are going to be assigned as weights to players divided into three sets---$A$, $B$, and $C$---in our proofs later on) as follows.
For some $t \in\mathbb{N}\setminus\{0\}$ such that
$10^{t} > 2^{\ceil{\log_2 n}+1}$ and
for $i \in \{1,\ldots,n\}$, define
\begin{eqnarray*}
a_i & = & 10^{t(m+1)+i}+\sum_{\substack{j\,:\ \textrm{clause $j$} \\ \textrm{ contains $x_i$}}} 10^{tj}
\text{ and } \\
b_i & = & 10^{t(m+1)+i}+\sum_{\substack{j\,:\ \textrm{clause $j$} \\ \textrm{ contains $\neg x_i$}}} 10^{tj},
\end{eqnarray*}
and for $j\in\{1,\ldots,m\}$ and $s\in\{0,\ldots,r\}$, define
\begin{eqnarray*}
c_{j,s} & = & 2^{s}\cdot 10^{tj}.
\end{eqnarray*}
Define the following three weight vectors:
\begin{eqnarray*}
W_A & = & (a_1,\ldots,a_k,b_1,\ldots,b_k), \\
W_B & = & (a_{k+1},\ldots,a_n,b_{k+1},\ldots,b_n), \text{ and}\\
W_C & = & (c_{1,0},\ldots,c_{m,r}).
\end{eqnarray*}	
Additionally, let
\begin{align*}
q'  ~=~ & \sum_{i=1}^{n} 10^{t(m+1)+i} + 2^{\ceil{\log_{2} n}}\sum_{j=1}^{m} 10^{tj}.
\end{align*}
\end{definition}
\noindent
Kaczmarek and Rothe~\cite{kac-rot:c:control-by-adding-players-to-change-or-maintain-shapley-shubik-or-the-penrose-banzhaf-power-index-in-WVGs-is-complete-for-np-pp,kac-rot:t:control-by-adding-players-to-change-or-maintain-shapley-shubik-or-the-penrose-banzhaf-power-index-in-WVGs-is-complete-for-np-pp} proved that each subset of $A\cup B\cup C$ whose total weight is equal to $q'$ uniquely corresponds to a truth assignment satisfying the boolean formula~$\phi$.
We are going to use this transformation and this fact in our proofs in the next section.

\section{%
	Complexity of Control by Deleting Players}
\label{sec:control-problems}

In this section, we present our results.
We start with the goal of decreasing the distinguished player's Penrose--Banzhaf index.
For this case, we provide a complete proof of Theorem~\ref{thm:deleting-decrease-PBI} below with all details 
(except all calculations which can be found in the appendix).
For the goals of nonincreasing or maintaining the distinguished player's Penrose--Banzhaf power, 
i.e., for upcoming Theorem~\ref{thm:deleting-remaining-PBI}, 
we defer the 
proofs to the appendix.

\begin{theorem}\label{thm:deleting-decrease-PBI}
	\textsc{Control-by-Deleting-Players-to-Decrease-PBI} is $\np^{\PP}$-complete.
\end{theorem}

\begin{proof}
	As mentioned above, the $\np^{\PP}$ upper bound has already been observed by Rey and Rothe~\cite{rey-rot:j:structural-control-in-weighted-voting-games}.

	To prove $\np^{\PP}$-hardness, we present a reduction from the
        $\np^{\PP}$-complete problem $\textsc{E-Minority-SAT}$.
        Let $(\phi, k)$ be a given instance of
	$\textsc{E-Minority-SAT}$, where $\phi$ is a Boolean formula in CNF
	with variables $x_1,\dots,x_n$ and $m$ clauses,
	and assume, without loss of generality, that $4\le k<n$.  

	Before we construct
	an instance of our control problem from $(\phi, k)$, we need to
	choose some numbers and introduce some notation.
	Let $z'$ be defined as in Table~\ref{tab:deleting-decrease-PBI}.
	Let
        \[
        z^{*}_1 = (k+2)z',
        \]
        and for $i\in\{2,\ldots,k\}$, define
	\[
	z^{*}_i = z^{*}_1 + \sum_{j=1}^{i-1}z^{*}_j.
	\]
	Let $t \in \mathbb{N}$ be such that 
	\begin{equation}
		10^{t} > \max\left\{2^{\ceil{\log_{2} n}+1}, 2z^{*}_k\right\}.
	\end{equation}
	For that $t$, let $A$, $B$, $C$, $a_i$, $b_i$, and $c_{j,s}$ for $i \in \{1,\ldots,n\}$, $j\in\{1,\ldots,m\}$, and $s\in\{0,\ldots,r\}$, and $q'$ be defined as in Definition~\ref{def:prereduction}.
	Finally, let
	\begin{eqnarray*}
		a_i '   & = & a_i \cdot 10^{t(m+1)+n},\\
		b_i '   & = & b_i \cdot 10^{t(m+1)+n}, \text{ and }\\
		c_{j,s}' & = & c_{j,s} \cdot 10^{t(m+1)+n},
	\end{eqnarray*}
	which together form the weight vector~$W_E$.
	
	Now, we construct an instance of our control problem $(\phi, k)$.
	Let $k$ be the limit for the
	number of players that can be deleted.
	Further, let 
	\[
	q'' = 10^{t(m+1)+n}q'.
	\]
	Note that,
	as shown by Kaczmarek and Rothe~\cite{kac-rot:c:control-by-adding-players-to-change-or-maintain-shapley-shubik-or-the-penrose-banzhaf-power-index-in-WVGs-is-complete-for-np-pp,kac-rot:t:control-by-adding-players-to-change-or-maintain-shapley-shubik-or-the-penrose-banzhaf-power-index-in-WVGs-is-complete-for-np-pp},
\begin{itemize}
  \item the first summand of $q'$ can be achieved only by
	taking exactly one element from $\{a_i,b_i\}$ for all
	$i\in\{1,\ldots,n\}$ and summing them up, and
  \item the second summand of $q'$ cannot be achieved only by the elements
	$c_{j,s}$, with $j\in\{1,\ldots,m\}$ and $s\in\{0,\ldots,r\}$ 
\end{itemize}
In the case of $q''$, the situation is analogous.

Define the quota of the weighted voting game by
\begin{equation}\label{deletingplayers-decreasePBI:def:q}
  q = 2\cdot \left(w_A + w_B  + w_C + w_E + 10^{t}\right) + 1.
\end{equation}

	\begin{table}[t!]
			\caption{\label{tab:deleting-decrease-PBI}
				Groups of players in the proof of
				Theorem~\ref{thm:deleting-decrease-PBI}, with their numbers and weights}
			\centering
				\begin{tabular}{c|c|c}
					\toprule
					\textbf{Group} & \textbf{Number of Players} & \textbf{Weights} \\
					\midrule
					player~$1$ & $1$ & $1$ \\
					\midrule
					$A$ & $2k$ & $W_A$ \\
					\midrule
					$B$ & $2n-2k$ & $W_B$ \\
					\midrule
					$C$ & $m(r+1)$ & $W_C$ \\
					\midrule
					$D$ & $k$ & \makecell{$q-q'-ix-x'$, \\  $i\in\{0,\ldots,k-1\}$} \\
					\midrule
					$E$ & $2n+m(r+1)$ & $W_E$ \\
					\midrule
					$F$ & $1$ & $q-q''-x'$ \\
					\midrule	
					$S$ & $k^2 (k+2)$ & \makecell{$q-(a_i + b_i)-jy-lz$, \\ %
						$i \in \{1,\ldots,k\}$, \\ $j\in\{0,\ldots,k+1\}$, \\ %
						$l\in\{1,\ldots,k\}$} \\
					\midrule
					$T$ & $k^2 (n+1)$ & \makecell{$q-(a_i + b_i)-jy'-lz'$, \\ %
						$i \in \{1,\ldots,k\}$, \\ $j\in\{0,\ldots,n\}$, \\ %
						$l\in\{1,\ldots,k\}$} \\
					\midrule
					$U$ & $k (k+2)$ & \makecell{$q-jy^{*}-z^{*}_i$, \\ %
						$i\in\{1,\ldots,k\}$, \\ %
						$j\in\{0,\ldots,k+1\}$} \\
					\midrule
					$V$ & $k (n+1)$ & \makecell{$q-jy^{**}-z^{*}_i$, \\ %
						$i\in\{1,\ldots,k\}$, \\ %
						$j\in\{0,\ldots,n\}$} \\
					\midrule
					$X$ & $k$ & $x=1$ \\
					\midrule
					$X'$ & $2k$ & $x'=k+1$ \\
					\midrule
					$Y$ & $k+1$ & $y=(2k+1)x'$ \\
					\midrule
					$Y'$ & $n$ & $y'=(k+2)y$ \\
					\midrule
					$Y^{*}$ & $k+1$ & $y^{*}=(n+1)y'$ \\
					\midrule
					$Y^{**}$ & $n$ & $y^{**}=(k+2)y^{*}$ \\
					\midrule
					$Z$ & $k+1$ & $z=(n+1)y^{**}$ \\
					\midrule
					$Z'$ & $k+1$ & $z'=(k+2)z$ \\
					\midrule
					$Z^{*}$ & $k$ & $z^{*}_1,\ldots,z^{*}_k$ \\
					\bottomrule
				\end{tabular}
	\end{table}
	Finally, let $N$ be the set of players in the game $\mathcal{G}$ divided into the groups and with the weights presented in Table~\ref{tab:deleting-decrease-PBI}.

	Let us first discuss which coalitions player~$1$ can be pivotal for
	in the game $\mathcal{G}$.\footnote{For any $M\subseteq N\setminus\{1\}$, player~$1$ is pivotal for some subset of the coalitions in $\mathcal{G}_{\setminus M}$.}
	
	Player~$1$ is pivotal for those coalitions of players in
	$N\setminus \{1\}$ whose total weight is $q-1$.
	Note that any two players from
	\begin{equation}
		\label{eq:union-of-sets}
		D\cup F \cup S \cup T \cup U \cup V
	\end{equation}
	together have a	weight larger than~$q$.
	Therefore, at most one player from the union of sets
	in Equation~(\ref{eq:union-of-sets})
	can be in any coalition player~$1$ is pivotal for.  Moreover,
	by Equation~(\ref{deletingplayers-decreasePBI:def:q}), all other players
	together have a total weight smaller than $q-1$. 
	That means that any coalition $K \subseteq N\setminus \{1\}$ with a total weight of $q-1$ has to contain \emph{exactly} one of
	the players in
	the union of sets from Equation~(\ref{eq:union-of-sets}).
	Now, whether this player is in $D$, $F$, $S$,
	$T$, $U$, or $V$
	has consequences as to which other players will also be in such a
	weight-$(q-1)$ coalition~$K$.
	Accordingly, we distinguish the following six cases:

	\begin{description}\label{cases}
		\item[Case 1:]\label{case-1} If $K$ contains a player from~$D$ (with weight, say,
		\[
		q-q'-ix-x'
		\]
		for some~$i$, $0 \leq i \leq k-1$), $K$ also has to
		contain those players from $A\cup B \cup C$ whose weights sum up to $q'$, $i$ players from $X$ with weight~$x=1$, and a player from~$X'$.
		Also, recall that $q'$ can be achieved only by a set
		of players whose weights take exactly one of the values from
		$\{a_{j},b_{j}\}$ for each $j\in\{1,\ldots,n\}$,
		so $K$ must contain exactly $n$ players from  
		$A \cup B$.
		\smallskip
		
		\item[Case 2:]\label{case-2} If $K$ contains the player from 
		  $F$ with weight
                  \[
                  q-q''-x',
                  \]
                  $K$ has to contain those players from $E$ whose weights sum up to $q''$ and a player from~$X'$.
		\smallskip
		
		\item[Case 3:]\label{case-3} If $K$ contains a player from $S$ (with weight, say,
		\[
		q-(a_i + b_i)-jy-lz
		\]
		for some~$i$, $1 \leq i \leq k$, some~$j$, $0 \le j \le k+1$, and some~$l$,     
		$1\le l \le k$), $K$ also contains some pair of players from~$A$, $j$ players from~$Y$, and $l$ players from~$Z$.
		\smallskip
		
		\item[Case 4:]\label{case-4} If $K$ contains a player from $T$ (with weight, say,
		\[
		q-(a_i + b_i)-jy'-lz'
		\]
		for some~$i$, $1 \leq i \leq k$, some~$j$, $0 \le j \le n$, and some~$l$,     
		$1\le l \le k$), $K$ also contains some pair of players from~$A$, $j$ players from~$Y'$, and $l$ players from~$Z'$.
		\smallskip
		
		\item[Case 5:]\label{case-5} If $K$ contains a player from $U$ (with weight, say,
		\[	  
		q-jy^{*}-z^{*}_i
		\]
		for some~$i$, $1\le i\le k$, and for some~$j$, $0 \le j \le k+1$), $K$ also contains $j$ players from $Y^{*}$ and a player from~$Z^{*}$.
		\smallskip
		
		\item[Case 6:]\label{case-6} If $K$ contains a player from $V$ (with weight, say,
		\[
		q-jy^{**}-z^{*}_i
		\]
		for some~$i$, $1\le i\le k$, and for some~$j$, $0 \le j \le n$), $K$ also contains $j$ players from $Y^{**}$ and a player from~$Z^{*}$.
	\end{description}

	For any $\np$ problem $L$ (represented by an NP machine~$M_L$) and input~$x$, let $\#L(x)$ denote the number of solutions of $M_L$ on input~$x$.
	For example, using the ``standard'' NP machine~$M_{\textsc{SAT}}$ that, given a boolean formula, guesses its truth assignments and verifies whether they are satisfying, $\#\textsc{SAT}(\phi)$ denotes the number of truth assignments satisfying~$\phi$.
	Let $\xi = \#\textsc{SAT}(\phi)$.

	Let \textsc{SubsetSum} (see, e.g.,~\cite{gar-joh:b:int}) be the following well-known, $\np$-complete problem:
\EP{SubsetSum}
   {Positive integer sizes $s_1, \ldots , s_d$ and a positive integer~$\alpha$.}
   {Is there a subset $H \subseteq \{1, \ldots , d\}$ such that
     \[
     \sum_{h \in H} s_h = \alpha?
     \]
     \vspace*{-3mm}
   }

	Using the result by Kaczmarek and Rothe~\cite{kac-rot:c:control-by-adding-players-to-change-or-maintain-shapley-shubik-or-the-penrose-banzhaf-power-index-in-WVGs-is-complete-for-np-pp,kac-rot:t:control-by-adding-players-to-change-or-maintain-shapley-shubik-or-the-penrose-banzhaf-power-index-in-WVGs-is-complete-for-np-pp},
	\begin{eqnarray*}
		\xi & = & \#\textsc{SubsetSum}(A\cup B\cup C, q')\\
		& = & \#\textsc{SubsetSum}(E,q'').
	\end{eqnarray*}
	
	Considering all the cases above, player~$1$'s Penrose--Banzhaf index in the game $\mathcal{G}$ yields to
	\begin{align*}
		\PenroseBanzhaf(\mathcal{G},1)
		~=~ & \frac{2k(2^k -1)\xi+ 2k\xi + k2^{k+1} (2^{k+1} - 2)}{2^{|N|-1}} \\
		& {}+ \frac{k2^n (2^{k+1} -2)+k2^{k+1} + k2^n}{2^{|N|-1}} \\
		~=~ & \frac{2k 2^k\xi + k (2^{n} + 2^{k+1})(2^{k+1} - 1)}{2^{|N|-1}}.
	\end{align*}

	We now show that $(\phi, k)\in\textsc{E-Minority-SAT}$ if and only if $(\mathcal{G}, 1, k)\in$
	\textsc{Control-by-Deleting-Players-to-Decrease-PBI}.

	\proofonlyif
	Suppose that $(\phi, k) \in \textsc{E-Minority-SAT}$,
	i.e., there exists an assignment of $x_1,\ldots,x_k$ such that at most half of the assignments of the remaining $n-k$ variables satisfy
	the Boolean formula~$\phi$.  Let us fix one of these
	satisfying assignments.
	From this fixed assignment, the first
	$k$ positions correspond uniquely to some players $A' \subseteq A$, $|A'|=|A\setminus A'|=k$.
	Let us remove the players in $A\setminus A'$ from the game~$\mathcal{G}$
	(see~\cite{kac-rot:c:control-by-adding-players-to-change-or-maintain-shapley-shubik-or-the-penrose-banzhaf-power-index-in-WVGs-is-complete-for-np-pp,kac-rot:t:control-by-adding-players-to-change-or-maintain-shapley-shubik-or-the-penrose-banzhaf-power-index-in-WVGs-is-complete-for-np-pp} for details).
	
	Since there are at most $2^{n-k-1}$ (i.e., fewer than $2^{n-k-1}+1$) assignments for
	$x_{n-k},\ldots,x_n$ which---together with the fixed assignments for
	$x_1,\ldots,x_k$---satisfy~$\phi$, there are fewer
	than $2^{n-k-1}+1$ subsets of $A\cup B \cup C$ such that the players' weights
	in each subset sum up to~$q'$. Player~$1$  is still pivotal for all coalitions described in Cases~$2$, $5$, and~$6$, for
	fewer than
	\[
        2k(2^k -1)(2^{n-k-1}+1)
        \]
        coalitions described in Case~$1$, and player~$1$ is not pivotal for any coalition described in Cases~$3$ and~$4$ anymore.
	Therefore,
		\begin{eqnarray*}
			\lefteqn{\PenroseBanzhaf(\mathcal{G}_{\setminus (A\setminus A')},1)}\\
			& < & \frac{2k\left(2^k -1\right)\left(2^{n-k-1}+1\right) + 2k\xi + k\left(2^{n}+2^{k+1}\right)}{2^{|N|-1-k}} \\
			& = & \PenroseBanzhaf(\mathcal{G},1),
		\end{eqnarray*}
	so player~$1$'s Penrose--Banzhaf index is strictly smaller in 
	$\mathcal{G}_{\setminus (A\setminus A')}$ than in~$\mathcal{G}$,
	i.e., we have constructed a yes-instance of our control problem.

	\proofif
	Assume now that $(\phi, k) \not\in \textsc{E-Minority-SAT}$,
	i.e., there does not exist any
	assignment of the variables $x_1,\ldots,x_k$ such that at most half
	of assignments of the remaining $n-k$ variables satisfy the Boolean
	formula~$\phi$.  In other words, for each assignment of $x_1,\ldots,x_k$,
	there exist at least
        \[
        2^{n-k-1}+1
        \]
        assignments of $x_{k+1},\ldots,x_n$
	that satisfy~$\phi$.
        This also means that 
	\begin{align}\label{xi}
		2^n & \ge  \xi \ge 2^k \left(2^{n-k-1}+1\right).
	\end{align}
	
	Let us start with the case considered in the previous implication, i.e., the case of deletion of $k$ players $A\setminus A'$ that correspond to some truth assignment for $\phi$ (but this time, we consider any of the assignments):
	\begin{eqnarray*}
		\lefteqn{\PenroseBanzhaf(\mathcal{G}_{\setminus (A\setminus A')},1)} \\
		& \ge & \frac{2k(2^k -1)(2^{n-k-1}+1) + 2k\xi + k(2^{n}+2^{k+1})}{2^{|N|-1-k}} \\
		& = & \PenroseBanzhaf(\mathcal{G},1). 
	\end{eqnarray*}
	
	Now, we are going to analyze other possible deletions.
	First, note that the number of coalitions from Case~$4$ is not smaller than the number of coalitions from any other case:
	\begin{eqnarray}
		\label{thelargest:3}
		k2^n (2^{k+1}-2) & \ge & k2^{k+1}(2^{k+1}-2), \\
		\label{thelargest:1or2}
		k2^n (2^{k+1}-2) %
		& \ge & 2k\xi (2^k - 1) > 2k\xi, \text{ and } \\
		\label{thelargest:5and6}
		k2^n (2^{k+1}-2) & \stackrel{k\ge 4}{>} & k(2^n + 2^n) \ge k(2^n + 2^{k+1}).
	\end{eqnarray}
	
	Let us consider now the summands defined by Cases~$5$ and~$6$.
	Removing
        \[
        l_1 + l_2 + l_3 = l < k
        \] 
        players from~$U$,\footnote{Note that if $l=k$, the index will increase by Equation~(\ref{thelargest:5and6}) and that the summands can decrease to zero only if we remove all $k$ players from $Z^{*}$.}
	$Y^{*}$, and~$Z^{*}$, respectively (i.e., those players forming coalitions from Case~$5$; the situation of the summand defined by Case~$6$ is analogous). Let us analyze first when the numerator decreases the most so that we do not have to consider all possible combinations for fixed~$l_1$, $l_2$, and~$l_3$: 
	After removing the $l_2$ players from $Y^{*}$ (and since they have the same weight, they are symmetric and it does not matter which ones we delete), it will be impossible for player~$1$ to be pivotal for the coalitions formed, i.e., by the players from $U$ that need more than
        \[
        k+1-l_2
        \]
        players from $Y^{*}$ to achieve their total weight of~$q-1$, so we do not need to delete these players from the game---it would not decrease player~$1$'s index any more.
        
	Also, the players from $U$ forming the most coalitions after the deletion of the mentioned players from $Y^{*}$ are players with weights of the form
		\[
		q - \left\lfloor \frac{k+1-l_2}{2} \right\rfloor y^{*}-z^{*}_i,
		\]
	        \hspace*{-1.8mm}$1\le i \le k$, each forming
                \[
                  {k+1-l_2 \choose \floor{\frac{k+1-l_2}{2}}}
                  \]
                  coalitions with weight $q-1$, so we get the biggest decrease if we remove $l_1$ of those players with $i$ such that the players with weight~$z^{*}_i$ appear in the new game.
                  
	Therefore, the summand defined by Case~$5$ changes to at least
	\begin{eqnarray*}
		& & \frac{(k-l_3)2^{k+1-l_2} - l_1{k+1-l_2 \choose \floor{\frac{k+1-l_2}{2}}}}{2^{|N|-1-l}} \\ 
		& \ge & \frac{k2^{k+1}}{2^{|N|-1}} + 2^{k+1}\frac{ \left(2^{l_1 + l_3}-1\right) k -  2^{l_1 + l_3 }\left(\frac{1}{2}l_1 + l_3\right)}{2^{|N|-1}}  
	\end{eqnarray*}
	and because
        \[
	2^{l_1 + l_3}  \ge l_1 + l_3 +1 \text{ for $l_1 + l_3 \ge 0$, and }
	k\ge l_1 + l_3 + 1,
	\]
	we have
	\[
	\left(2^{l_1 + l_3}-1\right)k-  2^{l_1 + l_3 }\left(\frac{1}{2}l_1 + l_3\right) \ge 0,
	\]
	so 
	\begin{align}\label{case:d}
		\frac{(k-l_3)2^{k+1-l_2} - l_1{k+1-l_2 \choose \floor{\frac{k+1-l_2}{2}}}}{2^{|N|-1-l}} &  \ge \frac{k2^{k+1}}{2^{|N|-1}}, 
	\end{align}
	i.e., the summand does not decrease (and if we also delete some $l'$ players from
        \[
        N\setminus (\{1\}\cup U \cup Y^{*} \cup Z^{*}),
        \]
        the summand's value will increase by a factor of~$2^{l'}$).

	Let us now consider Case~$3$.
	Let $l_1$, $l_2$, and $l_3$ with
        \[
        l_1+l_2+l_3 = l
        \]
        be the number of players being removed from~$S$, $Y$, and~$Z$, respectively.
	Let $l'<k$ (the case of $l'=k$ was already analyzed) be the number of players deleted from $A$ such that no ``whole pair'' $(a_i,b_i)$ for any $i\in\{1,\ldots,k\}$ is deleted, and let $1\le l+l' \le k$.
	Consider the case giving us the least number of coalitions for which $1$ stays pivotal.
	The players from $Y$ and $Z$ are symmetric, so it does not matter which ones we remove.
	In the new game, there are $k+1-l_2$ players with weight~$y$, $k+1-l_3$ players with weight~$z$, and $k-l'$ whole pairs in~$A$.
	Since player~$1$ cannot be pivotal for the coalitions that need more players from $Y$ or $Z$ than there are in the new game anymore, and for coalitions that need the removed players from $A$, we do not need to remove the players from $S$ forming coalitions that either need more players from $Y$ or $Z$ or the whole pair from $A$ that is not there anymore to achieve the total weight~$q-1$.
	The most coalitions (after deleting the $l_2 + l_3 + l'$ players) are formed by players from $S$ having weights of the form
	\begin{small}
		\[
		q-(a_i + b_i) - \left\lfloor \frac{k+1-l_2}{2} \right\rfloor y - \left\lceil \frac{k+1-l_3}{2} \right\rceil z
		\]
	\end{small}
	\hspace*{-2mm}(each forming
        \[
          {k+1-l_2 \choose \floor{\frac{k+1-l_2}{2}}}{k+1-l_3 \choose \ceil{\frac{k+1-l_3}{2}}}
          \]
          coalitions of weight \mbox{$q-1$}), and there are
	at least $k-l'\ge l_1$ players with such a weight.
	Let us remove any $l_1$ players with such a weight.
	From Case~$3$ in the new game, it can be seen that 
	\begin{small}
		\begin{align*}
			& \frac{(k-l')2^{k+1-l_2}\left(2^{k+1-l_3}-1\right) - l_1 {k+1-l_2 \choose \floor{\frac{k+1-l_2}{2}}}{k+1-l_3 \choose \ceil{\frac{k+1-l_3}{2}}}}{2^{|N|-1-l-l'}} \\
			\ge ~ & \frac{k2^{k+1}\left(2^{k+1}-2\right)}{2^{|N|-1}} \\
			& + \frac{(k-l')2^{2k + l_1 + l'} + k2^{k+2}  - (k-l')2^{k+1+l_1 + l_3 + l'}}{2^{|N|-1}}, \\
		\end{align*}
	\end{small}
	\hspace*{-0.7mm}and for $\left(2^{l_1 + l_3 + l'}-1\right)\frac{k2^{k+1}}{2^{|N|-1}}$ from the increase of the summand defined by Case~$5$, we have
	\begin{align*}
		& \frac{(k-l')2^{k+1-l_2}\left(2^{k+1-l_3}-1\right) - l_1 {k+1-l_2 \choose \floor{\frac{k+1-l_2}{2}}}{k+1-l_3 \choose \ceil{\frac{k+1-l_3}{2}}}}{2^{|N|-1-l-l'}} \\
		& + \left(2^{l_1 + l_3 + l'}-1\right)\frac{k2^{k+1}}{2^{|N|-1}} > ~ \frac{k2^{k+1}\left(2^{k+1}-2\right)}{2^{|N|-1}}.
	\end{align*}
	Note that if we removed some whole pair $(a_i,b_i)$, $i\in\{1,\ldots,k\}$, from $A$, we would get exactly the same number of coalitions for the summands as in the case of removing only one of the players from the pair, i.e., the new game would just contain fewer players, so the summands would even be larger (since the denominator would be smaller at the same time).
	Removing players from the rest of the groups will only even more increase the summands above.

	Let us consider removing
        \[
        l_1 + l_2 + l_3 + l' = l + l' \le k
        \]
        players from~$T$, $Y'$, $Z'$, and~$A$, respectively, assuming that $l'<k$ and no whole pair $(a_i,b_i)$ for any $i\in\{1,\ldots,k\}$ is removed from the game.
	Let us focus on the summand defined by Case~$4$: %
	It can be shown that the summand together with (a part of) the increase from the summand defined by Case~$6$,
        \[
        \left(2^{l + l'}-1\right)\frac{k2^n}{2^{|N|-1}},
        \]
        is greater than in~$\mathcal{G}$.
	However, removing these players can have an impact on the summand from Case~$1$: Unless $l_1=l_2=l_3=0$ and $l'>0$ (the latter is assumed because otherwise, we would make no changes among the players forming the coalitions counted in these cases), this summand alone increases; but in the other case, this summand changes to
	at least
	\begin{align*}
		& \frac{2k\left(2^k -1\right)2^{k-l'}\left(2^{n-k-1}+1\right)}{2^{|N|-1-l'}} \\
		= ~ & \frac{2k\left(2^k -1\right)2^{n}}{2^{|N|-1}} -  \frac{2k\left(2^k -1\right)\left(2^{n-1}-2^{k}\right)}{2^{|N|-1}} \\
		\stackrel{(\ref{xi})}{\ge} ~ & \frac{2k\left(2^k -1\right)\xi}{2^{|N|-1}} -  \frac{2k\left(2^k -1\right)\left(2^{n-1}-2^{k}\right)}{2^{|N|-1}}.
	\end{align*}
	It also can be shown that in this case, the increase of the summand from Case~$4$ is greater
	than the subtrahend above.
	Moreover, if we deleted at least one more player from $A$ from a pair $(a_i, b_i)$ from which we have already removed a player in the cases above or any other player that forms some coalitions from Case~$1$ causing the decrease of the summand to~$0$, it will not change the number of coalitions counted in the other summand, but will decrease its denominator and therefore, by Equation~(\ref{thelargest:1or2}), the sum will still increase.
	
	The summands defined by Cases~$1$ and~$2$ might collapse to $0$ also by deleting one player
	from each case.	
	But then, 	 
	if we also delete a player (or players) decreasing   
	the summands defined by Cases~$1$ and~$2$ to~$0$---noting that there have to be removed at least two players to decrease both the summands---the increase of the summand from Case~$5$ is greater than them by Equation~(\ref{thelargest:1or2}).  
	
	Summarizing the above analysis:
	\begin{enumerate}
		\item Removing $k$ players as it was in the previous implication but for any set corresponding to a truth assignment for the first $k$ variables, the index increases in the new game created from $\mathcal{G}$.
		\item If we remove some of the players forming coalitions from Cases~$5$ and~$6$, there are two possibilities:
		\begin{itemize}
			\item Either we remove $k$ players only from~$Z^{*}$, excluding all the coalitions from the index in the worst situation; however, by Equation~(\ref{thelargest:5and6}), the new index is greater due to the coalitions from Case~$4$,
			\item or we remove fewer than $k$ players---deleting or not other players forming coalitions from the other cases---then, the summands defined by Cases~$5$ and~$6$ increase
			(each nondecreases, but either no other players were removed and the other summands increase, or some other player was removed, then this summand increases).	
		\end{itemize}
		\item Removing players forming the coalitions from Case~$3$ makes the sum of this summand together with the increase of the summand defined by Case~$5$ in the new game greater.
		\item If we remove some players forming the coalitions from Case~$4$, this summand (possibly
		together with the summand from Case~$6$) increases, and when the summand from Case~$1$ decreases (assuming we remove only the players that are part of the coalitions from Case~$4$), the sum of the two summands still increases.
		\item If we remove the players that are crucial for the summands from Cases~$1$ and~$2$ (not studied in the previous situation) and they decrease even to~$0$, the index increases due to Equation~(\ref{thelargest:1or2})
		and the fact that the summand by Case~$6$ increases.
	\end{enumerate}
	
	When removing other combinations of players, the Penrose--Banzhaf index of player~$1$ cannot be smaller than in the cases we presented above.
	Thus the Penrose--Banzhaf index of player~$1$
	cannot decrease for $\mathcal{G}$ by deleting up to $k$ players from $N\setminus\{1\}$ in game~$\mathcal{G}$, so we have
	a no-instance of our control problem.~\end{proof}

As mentioned earlier, the proof of $\np^{\PP}$-completeness for the other two problems
can be found in the appendix.
\begin{theorem}\label{thm:deleting-remaining-PBI}
		The problems
		\textsc{Control-by-Deleting-Players-to-Nonincrease-PBI} and
		\textsc{Control-by-Deleting-Players-to-Maintain-PBI}
		are $\np^{\PP}$-complete. 
\end{theorem}

\section{Conclusions}
\label{sec:conclusions}

We have shown that control by deleting players from WVGs with the goals of decreasing, nonincreasing, and maintaining a given player's Penrose--Banzhaf index is $\np^{\PP}$-complete.
We have employed the proof method of Kaczmarek and Rothe~\cite{kac-rot:c:control-by-adding-players-to-change-or-maintain-shapley-shubik-or-the-penrose-banzhaf-power-index-in-WVGs-is-complete-for-np-pp,kac-rot:t:control-by-adding-players-to-change-or-maintain-shapley-shubik-or-the-penrose-banzhaf-power-index-in-WVGs-is-complete-for-np-pp}, but our proofs are considerably more involved.
As a challenging task for future work, we propose to further develop this method
to solve the missing cases: the goals of increasing or nondecreasing the distinguished player's power for the Penrose--Banzhaf index and all five goals for the Shapley--Shubik index and other power indices. 
For instance, the normalized Penrose--Banzhaf index poses
greater technical challenge in the comparison to its alternative: It requires the analysis of coalitions, not only
those for which a distinguished player is pivotal, but we need the information about all players in a given game.  

We also propose to study these problems in the model of Kaczmarek and
Rothe~\cite{kac-rot:j:controlling-weighted-voting-games-by-deleting-or-adding-players-with-or-without-changing-the-quota}
where the quota is changing proportionally when players are removed.
Finally, this method might also help to close the huge complexity gap between a known $\np^{\PP}$ upper bound and much weaker
known lower bounds for other problems on WVGs (e.g., for false-name manipulation~\cite{azi-bac-elk-pat:j:false-name-manipulations-wvg,rey-rot:j:false-name-manipulation-PP-hard}), 
as well as on graph-restricted WVGs that are controlled by adding or deleting edges in the underlying communication network~\cite{kac-rot-tal:j:control-by-adding-or-deleting-edges-in-grwvgs}.

\begin{ack}
  We thank the reviewers for helpful comments; in particular, we thank the reviewer who suggested how to convey certain complexity issues more completely to the reader, which we were happy to follow to the extent permitted by limitation of space.
  This work was supported in part by DFG grant
  RO~1202/21-2 (project number 438204498).
\end{ack}

\bibliography{joannabibliography}

\clearpage

\appendix

\section{Full Proof of Theorem~\ref{thm:deleting-decrease-PBI} and of the First Part of Theorem~\ref{thm:deleting-remaining-PBI}}

\begin{proofs} [Full proof of $\np^{\PP}$-hardness of \textsc{Control-by-Deleting-Players-to-Decrease-PBI} (Theorem~\ref{thm:deleting-decrease-PBI}) and of \textsc{Control-by-Deleting-Players-to-Nonin\-crease-PBI} (first part of Theorem~\ref{thm:deleting-remaining-PBI}).]
	
	We will prove $\np^{\PP}$-hardness using a reduction from
	$\textsc{E-Minority-SAT}$.  Let $(\phi, k)$ be a given instance of
	$\textsc{E-Minority-SAT}$, where $\phi$ is a Boolean formula in CNF
	with variables $x_1,\dots,x_n$ and $m$ clauses,
	and $4\le k<n$.  
	Before we construct
	instances of our control problems from $(\phi, k)$, we need to
	choose some numbers and introduce some notation.

	Let $z^{*}_1 = (k+2)z'$ for $z'$ defined in 
	Table~\ref{tab:deleting-decrease-PBI}, and define
	\[
	z^{*}_i = z^{*}_1 + \sum_{j=1}^{i-1}z^{*}_j
	\]
	for $i\in\{2,\ldots,k\}$. 
	Let $t \in \mathbb{N}$ be such that 
	\begin{equation}
		10^{t} > \max\left\{2^{\ceil{\log_{2} n}+1}, 2z^{*}_k\right\}.
	\end{equation}
	For that $t$, let $A$, $B$, $C$, $a_i$, $b_i$, and $c_{j,s}$ for $i \in \{1,\ldots,n\}$, $j\in\{1,\ldots,m\}$, and $s\in\{0,\ldots,r\}$, and $q'$ be defined as in Definition~\ref{def:prereduction}.
	Finally, let
	\begin{eqnarray*}
		a_i '   & = & a_i \cdot 10^{t(m+1)+n},\\
		b_i '   & = & b_i \cdot 10^{t(m+1)+n}, \text{ and }\\
		c_{j,s}' & = & c_{j,s} \cdot 10^{t(m+1)+n},
	\end{eqnarray*}
	which together form the weight vector~$W_E$.
	
	Now, we construct
	the instances of our control problems.  Let $k$ be the limit for the
	number of players that can be deleted.
	Further, let 
	\[
	q'' = 10^{t(m+1)+n}q'.
	\]
	Note that,
	as shown by Kaczmarek and Rothe~\cite{kac-rot:c:control-by-adding-players-to-change-or-maintain-shapley-shubik-or-the-penrose-banzhaf-power-index-in-WVGs-is-complete-for-np-pp,kac-rot:t:control-by-adding-players-to-change-or-maintain-shapley-shubik-or-the-penrose-banzhaf-power-index-in-WVGs-is-complete-for-np-pp},
	the first summand of $q'$ can be achieved only by
	taking exactly one element from $\{a_i,b_i\}$ for all
	$i\in\{1,\ldots,n\}$ and summing them up, and that
	the second summand of $q'$ cannot be achieved only by the elements
	$c_{j,s}$, with $j\in\{1,\ldots,m\}$ and $s\in\{0,\ldots,r\}$ 	
	(in the case of $q''$ and the values in $W_E$, the situation is analogous).

        From our given instance $(\phi, k)$ of $\textsc{E-Minority-SAT}$, we now construct an instance $(\mathcal{G},1,k)$ of \textsc{Control-by-Deleting-Players-to-Decrease-PBI} and an instance $(\mathcal{G}',1,k)$ of \textsc{Control-by-Deleting-Players-to-Nonincrease-PBI}, i.e., the distinguished player in both games is~$1$, and at most $k$ players can be deleted from them.
        Define the quota of both weighted voting games, $\mathcal{G}$ and $\mathcal{G}'$, by
	\begin{equation}\label{addingplayers-increasePBI:def:q}
		q=2\cdot \left(w_A + w_B  + w_C + w_E + 10^{t}\right) +1.
	\end{equation}
        Let $N$ be the set of players in  
	game~$\mathcal{G}$, divided into the groups and with the weights presented in Table~\ref{tab:deleting-decrease-PBI}.
	Finally, the set of players in 
	game~$\mathcal{G}'$ is defined by $N'=N\setminus (S\cup U \cup Y \cup Y^{*} \cup Z)$.

	Let us first discuss which coalitions player~$1$ can be pivotal for
	in game $\mathcal{G}$.\footnote{For any $M\subseteq N\setminus\{1\}$, $1$ is pivotal for some subset of the coalitions in $\mathcal{G}_{\setminus M}$. Also note that $\mathcal{G}' = \mathcal{G}_{\setminus (N\setminus N')}$.}
	Player~$1$ is pivotal for those coalitions of players in
	$N\setminus \{1\}$ whose total weight is $q-1$.
	First, note that any two players from $D\cup 
	F \cup S \cup T \cup U \cup V$
	together have a
	weight larger than $q$.  Therefore, at most one player from $D\cup 
	F\cup S \cup T \cup U \cup V$
	can be in any coalition player~$1$ is pivotal for.  Moreover,
	by~(\ref{addingplayers-increasePBI:def:q}), all other players
	together have a total weight smaller than $q-1$. 
	That means that any coalition $K \subseteq N\setminus \{1\}$ with a total weight of $q-1$ has to contain \emph{exactly} one of
	the players in $D\cup 
	F \cup S \cup T \cup U \cup V$.  Now, whether this player is in $D$, $F$, $S$,
	$T$, 
	$U$, or $V$
	has consequences as to which other players will also be in such a
	weight-$(q-1)$ coalition~$K$:

	\begin{description}
		\item[Case 1:] If $K$ contains a player from~$D$ (with weight, say,
		\[
		q-q'-ix-x'
		\]
		for some~$i$, $0 \leq i \leq k-1$), $K$ also has to
		contain those players from $A\cup B \cup C$ whose weights sum up to $q'$, $i$ players from $X$ with weight~$x=1$, and a player from~$X'$.
		Also, recall that $q'$ can be achieved only by a set
		of players whose weights take exactly one of the values from
		$\{a_{i'},b_{i'}\}$ for each $i'\in\{1,\ldots,n\}$,
		so $K$ must contain exactly $n$ players from  
		$A \cup B$.
		\smallskip
		
		\item[Case 2:] If $K$ contains the player from 
		$F$ with weight
		\[
		q-q''-x',
		\]
		$K$ has to contain those players from $E$ whose weights sum up to $q''$ and a player from~$X'$.
		\smallskip
		
		\item[Case 3:] If $K$ contains a player from $S$ (with weight, say,
		\[
		q-(a_i + b_i)-jy-lz
		\]
		for some~$i$, $1 \leq i \leq k$, some~$j$, $0 \le j \le k+1$, and some~$l$,     
		$1\le l \le k$), $K$ also contains some pair of players from~$A$, $j$ players from~$Y$, and $l$ players from~$Z$.
		\smallskip
		
		\item[Case 4:] If $K$ contains a player from $T$ (with weight, say,
		\[
		q-(a_i + b_i)-jy'-lz'
		\]
		for some~$i$, $1 \leq i \leq k$, some~$j$, $0 \le j \le n$, and some~$l$,     
		$1\le l \le k$), $K$ also contains some pair of players from~$A$, $j$ players from~$Y'$, and $l$ players from~$Z'$.
		\smallskip
		
		\item[Case 5:] If $K$ contains a player from $U$ (with weight, say,
		\[	  
		q-jy^{*}-z^{*}_i
		\]
		for some~$i$, $1\le i\le k$, and for some~$j$, $0 \le j \le k+1$), $K$ also contains $j$ players from $Y^{*}$ and a player from~$Z^{*}$.
		\smallskip
		
		\item[Case 6:] If $K$ contains a player from $V$ (with weight, say,
		\[
		q-jy^{**}-z^{*}_i
		\]
		for some~$i$, $1\le i\le k$, and for some~$j$, $0 \le j \le n$), $K$ also contains $j$ players from $Y^{**}$ and a player from~$Z^{*}$.
	\end{description}

	For any $\np$ problem $L$ (represented by an NP machine~$M_L$) and input~$x$, let $\#L(x)$ denote the number of solutions of $M_L$ on input~$x$.
	For example, using the ``standard'' NP machine~$M_{\textsc{SAT}}$ that, given a boolean formula, guesses its truth assignments and verifies whether they are satisfying, $\#\textsc{SAT}(\phi)$ denotes the number of truth assignments satisfying~$\phi$.
	Let $\xi = \#\textsc{SAT}(\phi)$.
	
	Let \textsc{SubsetSum} (see, e.g.,~\cite{gar-joh:b:int}) be the well-known, $\np$-complete problem: 
	\EP{SubsetSum}
	{Positive integer sizes $s_1, \ldots , s_d$ and a positive integer~$\alpha$.}
	{Is there a subset $H \subseteq \{1, \ldots , d\}$ such that
		\[
		\sum_{h \in H} s_h = \alpha?
		\]
		\vspace*{-3mm}
	}
	Using a result by Kaczmarek and Rothe~\cite{kac-rot:c:control-by-adding-players-to-change-or-maintain-shapley-shubik-or-the-penrose-banzhaf-power-index-in-WVGs-is-complete-for-np-pp,kac-rot:t:control-by-adding-players-to-change-or-maintain-shapley-shubik-or-the-penrose-banzhaf-power-index-in-WVGs-is-complete-for-np-pp},
	\begin{eqnarray*}
		\xi & = & \#\textsc{SubsetSum}(A\cup B\cup C, q')\\
		& = & \#\textsc{SubsetSum}(E,q'').
	\end{eqnarray*}  
	
	Considering all the cases above, player~$1$'s Penrose--Banzhaf indices in game $\mathcal{G}$ and in game $\mathcal{G}'$ are as follows:
	\begin{align*}
		\PenroseBanzhaf(\mathcal{G},1) & =\frac{2k(2^k -1)\xi+ 2k\xi + k2^{k+1} (2^{k+1} - 2)}{2^{|N|-1}} \\
		& +\frac{ k2^n (2^{k+1} -2)+k2^{k+1} + k2^n}{2^{|N|-1}} \\
		& =\frac{2k 2^k\xi + k (2^{n} + 2^{k+1})(2^{k+1} - 1)}{2^{|N|-1}} ~ \textrm{ and } \\
		\PenroseBanzhaf(\mathcal{G}',1) & =\frac{2k(2^k -1)\xi+ 2k\xi  + k2^n (2^{k+1} -2) + k2^n}{2^{|N'|-1}} \\
		& =\frac{2k 2^k\xi + k 2^{n}(2^{k+1} - 1)}{2^{|N'|-1}}.
	\end{align*}

	We now show that the following statements are true, which shows the correctness of our reductions: 
	\begin{enumerate}
		\item  $(\phi, k)$ is a
		yes-instance of $\textsc{E-Minority-SAT}$ if and only if $(\mathcal{G}, 1, k)$ as defined above is a yes-instance of
		\textsc{Control-by-Deleting-Players-to-Decrease-PBI},
		\item  $(\phi, k)$ is a
		yes-instance of $\textsc{E-Minority-SAT}$ if and only if $(\mathcal{G}', 1, k)$ as defined above is a yes-instance of
		\textsc{Control-by-Deleting-Players-to-Nonincrease-PBI}.
	\end{enumerate}

	\proofonlyif
	Suppose that $(\phi, k)$ is a yes-instance of $\textsc{E-Minority-SAT}$,
	i.e., there exists an assignment of $x_1,\ldots,x_k$ such that at most half of assignments of the remaining $n-k$ variables satisfies
	the Boolean formula~$\phi$.  Let us fix one of these
	satisfying assignments.  From this fixed assignment, the first
	$k$ positions correspond to the players $A' \subseteq A$, $|A'|=|A\setminus A'|=k$,
	for which the players $A\setminus A'$ are removed from game~$\mathcal{G}$ and from game~$\mathcal{G}'$
	(see~\cite{kac-rot:c:control-by-adding-players-to-change-or-maintain-shapley-shubik-or-the-penrose-banzhaf-power-index-in-WVGs-is-complete-for-np-pp,kac-rot:t:control-by-adding-players-to-change-or-maintain-shapley-shubik-or-the-penrose-banzhaf-power-index-in-WVGs-is-complete-for-np-pp} for details).
	
	Since there are at most $2^{n-k-1}$, i.e., less than $2^{n-k-1}+1$, assignments for
	$x_{n-k},\ldots,x_n$ which---together with the fixed assignments for
	$x_1,\ldots,x_k$---satisfy~$\phi$, there are less
	than $2^{n-k-1}+1$ subsets of $A\cup B \cup C$ such that the players' weights
	in each subset sum up to~$q'$. Player~$1$  is still pivotal for all coalitions described in Case~$2$, Case~$5$ and Case~$6$, 
	for less than $2k(2^k -1)(2^{n-k-1}+1)$, coalitions described in Case~$1$, and is not pivotal for any coalition described in Case~$3$ and Case~$4$ anymore. Therefore,
	\begin{small}	
		\begin{eqnarray*}
		  \lefteqn{\PenroseBanzhaf(\mathcal{G}_{\setminus (A\setminus A')},1)}\\
                  & < & \frac{2k(2^k -1)(2^{n-k-1}+1) + 2k\xi + k(2^{n}+2^{k+1})}{2^{|N|-1-k}} \\
	          & = & \frac{k(2^k -1)(2^{n-k}+2)2^k + 2k \xi 2^k +k(2^n + 2^{k+1})2^k}{2^{|N|-1}} \\
		  & = & \frac{2k 2^k \xi + k(2^{n}+2^{k+1})(2^{k+1} -1)}{2^{|N|-1}}  = \PenroseBanzhaf(\mathcal{G},1) ~ \textrm{ and }\\
		  \lefteqn{\PenroseBanzhaf(\mathcal{G}_{\setminus (A\setminus A')}',1)} \\
                  & \le & \frac{2k(2^k -1)2^{n-k-1} + 2k\xi + k2^{n}}{2^{|N|-1-k}} \\
		  & = & \frac{k(2^k -1)2^{n-k}2^k + 2k \xi 2^k +k2^n 2^k}{2^{|N|-1}} \\
		  & = & \frac{2k 2^k \xi + k2^{n}(2^{k+1} -1)}{2^{|N|-1}}  = \PenroseBanzhaf(\mathcal{G}',1), 
		\end{eqnarray*}
	\end{small}
	\hspace*{-2.5mm}
	so player~$1$'s Penrose--Banzhaf power index is strictly smaller in 
	game~$\mathcal{G}_{\setminus (A\setminus A')}$ than in~$\mathcal{G}$, and is not greater in
	game~$\mathcal{G}_{\setminus (A\setminus A')}'$ than in~$\mathcal{G}'$,
	i.e., we have constructed yes-instances of our two control problems.

	\proofif
	Let us
	assume now that $(\phi, k)$ is a no-instance of $\textsc{E-Minority-SAT}$,
	i.e., there does not exist any
	assignment of the variables $x_1,\ldots,x_k$ such that at most half
	of assignments of the remaining $n-k$ variables satisfies the Boolean
	formula~$\phi$.  In other words, for each assignment of $x_1,\ldots,x_k$,
	there exist at least $2^{n-k-1}+1$ assignments of $x_{k+1},\ldots,x_n$
	that satisfy~$\phi$. This also means that 
	\begin{align}\label{app:xi}
		2^n & \ge  \xi \ge 2^k \left(2^{n-k-1}+1\right).
	\end{align}
	
	Let us start with the case considered in the
        \emph{``Only if''} direction of the proof above, i.e., the case of deleting $k$ players $A'$ that correspond to some truth assignment for $\phi$ (but this time, we consider any of the assignments).
        When deleting $A'$ from $N$ in~$\mathcal{G}$, we obtain:
	\begin{small}	
		\begin{eqnarray*}
		  \lefteqn{\PenroseBanzhaf(\mathcal{G}_{\setminus A'},1)} \\
                  & \ge & \frac{2k(2^k -1)(2^{n-k-1}+1) + 2k\xi + k(2^{n}+2^{k+1})}{2^{|N|-1-k}} \\
		  & = & \frac{k(2^k -1)(2^{n-k}+2)2^k + 2k \xi 2^k +k(2^n + 2^{k+1})2^k}{2^{|N|-1}} \\
		  & = & \frac{2k 2^k \xi + k(2^{n}+2^{k+1})(2^{k+1} -1)}{2^{|N|-1}}  \\
                  & = & \PenroseBanzhaf(\mathcal{G},1),
		\end{eqnarray*}
                and when deleting $A'$ from $N'$ in~$\mathcal{G}'$, we obtain:
		\begin{eqnarray*}
		  \lefteqn{\PenroseBanzhaf(\mathcal{G}_{\setminus A'}',1)} \\
                  & > & \frac{2k(2^k -1)2^{n-k-1} + 2k\xi + k2^{n}}{2^{|N|-1-k}} \\
		  & = & \frac{k(2^k -1)2^{n-k}2^k + 2k \xi 2^k +k2^n 2^k}{2^{|N|-1}} \\
		  & = & \frac{2k 2^k \xi + k2^{n}(2^{k+1} -1)}{2^{|N|-1}} \\
                  & = & \PenroseBanzhaf(\mathcal{G}',1). 
		\end{eqnarray*}
	\end{small}
	Thus we have constructed no-instances of our two control problems in the case we currently consider.

	Now, we are going to analyze
        all other possible deletions. 
	First note that 
	the number of coalitions from Case~$4$ is not smaller than the number of coalitions from any other case:
  \begin{eqnarray}
    \label{app:thelargest:3}
	k2^n \left(2^{k+1}-2\right) & \ge & k2^{k+1}\left(2^{k+1}-2\right),
	\\
        \nonumber
        k2^n \left(2^{k+1}-2\right) & = & 2k2^{n}\left(2^{k}-1\right) \\
                                   & \ge & 2k\xi \left(2^k - 1\right) ~>~ 2k\xi, \text{ and }
        \label{app:thelargest:1or2}
	\\
        \label{app:thelargest:5and6}
	k2^n \left(2^{k+1}-2\right) &
	\stackrel{k\ge 4}{>} & k\left(2^n + 2^n\right) ~\ge~ k\left(2^n + 2^{k+1}\right).
	\end{eqnarray}
      (\ref{app:thelargest:3}) compares the number of coalitions from Case~$4$ with those from from Case~$3$; (\ref{app:thelargest:1or2}) compares Case~$4$ with Cases~$1$ and~$2$; and (\ref{app:thelargest:5and6}) compares Case~$4$ with Cases~$5$ and~$6$.
      
	Consider the summands defined by Cases~$5$ and~$6$.
	Suppose we remove $l_1 + l_2 + l_3 = l < k$ players from~$U$, $Y^{*}$, and $Z^{*}$, respectively,\footnote{Note that if $l=k$, the index will increase by~(\ref{app:thelargest:5and6}), and that the summands can decrease to zero only if we remove all $k$ players from~$Z^{*}$.}
	i.e., players forming coalitions from Case~$5$ (the situation of the summand defined by Case~$6$ is analogous).
	Let us analyze first when the numerator decreases the most so that we do not have to consider all possible combinations for fixed $l_1$, $l_2$, and $l_3$: 
	After removing the $l_2$ players from $Y^{*}$ (and since they have the same weight, they are symmetric and it does not matter which ones we delete), it will be impossible for player~$1$ to be pivotal for the coalitions formed, i.a., by the players from $U$ that need more than $k+1-l_2$ players from $Y^{*}$ to achieve their total weight of~$q-1$, so we do not need to delete these players from the game---it would not decrease player~$1$'s index
	any more. Also, the players from $U$ forming at most number of coalitions after the deletion of the mentioned players from $Y^{*}$ are players with weight of form
	$$q-\left\lfloor \frac{k+1-l_2}{2} \right\rfloor y^{*}-z^{*}_i,$$ 
	$1\le i \le k$, each forming
	$${k+1-l_2 \choose \floor{\frac{k+1-l_2}{2}}}$$ 
	coalitions with weight~$q-1$, so we get the largest decrease if we remove $l_1$ from those players with $i$ such that the players with weight~$z^{*}_i$ appear in the new game	
	(note that $l_1 < k - l_3$). Therefore, the summand defined by Case~$5$ changes to at least 
	\begin{small}	
		\begin{align*}
			& \frac{(k-l_3)2^{k+1-l_2} - l_1{k+1-l_2 \choose \floor{\frac{k+1-l_2}{2}}}}{2^{|N|-1-l}} \\ 
			& =  \frac{2^l \Big((k-l_3)2^{k+1-l_2} - l_1{k+1-l_2 \choose \floor{\frac{k+1-l_2}{2}}}\Big)}{2^{|N|-1}} \\
			& = \frac{2^{l_1 + l_3}(k-l_3)2^{k+1} - 2^{l}l_1{k+1-l_2 \choose \floor{\frac{k+1-l_2}{2}}}}{2^{|N|-1}} \\
			& = \frac{k2^{k+1} + ((2^{l_1 + l_3}-1)k - 2^{l_1 + l_3}l_3) 2^{k+1} - 2^{l}l_1{k+1-l_2 \choose \floor{\frac{k+1-l_2}{2}}}}{2^{|N|-1}} \\
			& \ge \frac{k2^{k+1}}{2^{|N|-1}} + \frac{ (2^{l_1 + l_3}-1)k2^{k+1}- 2^{k+1+l_1 + l_3}l_3 - 2^{l}l_1 2^{k-l_2}}{2^{|N|-1}} \\
			& = \frac{k2^{k+1}}{2^{|N|-1}} + \frac{ (2^{l_1 + l_3}-1)k2^{k+1}- 2^{k+1+l_1 + l_3}l_3 - l_1 2^{k+l_1 + l_3}}{2^{|N|-1}} \\
			& = \frac{k2^{k+1}}{2^{|N|-1}} + 2^{k+1}\frac{ (2^{l_1 + l_3}-1)k-  2^{l_1 + l_3 }(\frac{1}{2}l_1 + l_3)}{2^{|N|-1}}  
		\end{align*}
	\end{small}
	and because $2^{l_1 + l_3}  \ge l_1 + l_3 +1$ for $l_1 + l_3 \ge 0$, and $k\ge l_1 + l_3 + 1$,
	\begin{align*}
		(2^{l_1 + l_3}-1)k & -  2^{l_1 + l_3 }(\frac{1}{2}l_1 + l_3)  \\ & \ge (2^{l_1 + l_3}-1)(l_1+l_3+1)-  2^{l_1 + l_3 }(l_1 + l_3) \\
		& = 2^{l_1 + l_3}-1 -l_1-l_3 \ge 0,
	\end{align*}
	so 
	\begin{align}\label{app:case:d}
		\frac{(k-l_3)2^{k+1-l_2} - l_1{k+1-l_2 \choose \floor{\frac{k+1-l_2}{2}}}}{2^{|N|-1-l}} &  \ge \frac{k2^{k+1}}{2^{|N|-1}}, 
	\end{align}
	i.e., the summand does not decrease (and if we delete also some $l'$ players from $N\setminus (\{1\}\cup U \cup Y^{*} \cup Z^{*})$, the value of the summand will increase by a factor of~$2^{l'}$).

	Now, let us consider only Case~$3$ (which is counted only for $\mathcal{G}$).
	Let $l_1,l_2,l_3$ with $l_1+l_2+l_3 = l$ be the number of players being removed from $S$, $Y$, and $Z$, respectively, let $l'<k$ (the case of $l'=k$ was already analyzed) be the number of players deleted from $A$ such that no whole pair $(a_i,b_i)$ for any $i\in\{1,\ldots,k\}$ is deleted and let $1\le l+l' \le k$. Let us consider the case giving us 
	the smallest number of coalitions for which $1$ stays pivotal. The players from $Y$ and $Z$ are symmetric, therefore it does not matter which ones we remove. In the new game, there are $k+1-l_2$ players with weight~$y$, $k+1-l_3$ players with weight~$z$, and $k-l'$ whole pairs in $A$. Since player~$1$ cannot be pivotal for the coalitions that need more players from $Y$ or $Z$ than there are in the new game anymore, and for coalitions that need the removed players from $A$, we do not need to remove the players from $S$ forming coalitions that either need more players from $Y$ or $Z$ or the whole pair from $A$ that is not there anymore to achieve the total weight~$q-1$. The most number of coalitions (after deleting the $l_2 + l_3 + l'$ players) are formed by players from $S$ having weight of form	
	\[
	q-(a_i + b_i) - \left\lfloor \frac{k+1-l_2}{2}\right\rfloor y - \left\lceil\frac{k+1-l_3}{2}\right\rceil z
	\]
	(each forming 
	$${k+1-l_2 \choose \floor{\frac{k+1-l_2}{2}}}{k+1-l_3 \choose \ceil{\frac{k+1-l_3}{2}}}$$ 
	coalitions of weight $q-1$), and there are not less than $k-l'\ge l_1$ players with such a weight. Let us remove any $l_1$ players with such a weight. From Case~$3$ in the new game, we get the summand of at least the following value in player~$1$'s index---let us consider it dividing it into four cases: 
	\begin{enumerate}
		\item $l_3 > 0 \wedge l_1 = 0$ :
		\begin{align*}
			& \frac{(k-l')2^{k+1-l_2}(2^{k+1-l_3}-1)}{2^{|N|-1-l-l'}} \\  = ~ & \frac{2^{l+l'}(k-l')2^{k+1-l_2}(2^{k+1-l_3}-1)}{2^{|N|-1}} \\
			= ~ & \frac{2^{l'}(k-l')2^{k+1}(2^{k+1}-2^{l_3})}{2^{|N|-1}}  
		\end{align*}
		and since $2^{l'}\ge l'+1$ for $l'\ge 0$, and $l'<k$,
		\begin{align*}
			(2^{l'}-1)k & \ge (2^{l'}-1)(l'+1) \\
			& = 2^{l'}l' + 2^{l'} - l' - 1 \\
			& \ge 2^{l'} l', 
		\end{align*}
		which implies
		\begin{align}\label{k:inequality}
			2^{l'}(k-l') & \ge k, 
		\end{align} 
		therefore,
		\begin{small}	
			\begin{align*}
				& \frac{(k-l')2^{k+1-l_2}(2^{k+1-l_3}-1)}{2^{|N|-1-l-l'}} \\ \ge ~ & 
				\frac{k2^{k+1}(2^{k+1}-2^{l_3})}{2^{|N|-1}}  \\
				= ~ & \frac{k2^{k+1}(2^{k+1}-2)}{2^{|N|-1}} - \frac{k2^{k+1}(2^{l_3}-2)}{2^{|N|-1}} \\
				= ~ & \frac{k2^{k+1}(2^{k+1}-2)}{2^{|N|-1}} - 2^{l_3}\frac{k(2^{k+1}-2)}{2^{|N|-1}} - \frac{2^{l_3+1}k}{2^{|N|-1}} + \frac{k2^{k+2}}{2^{|N|-1}}.
			\end{align*}
		\end{small}
		From (\ref{app:case:d}), there is an increase of the summand defined by Case~$5$ by at least $\alpha=(2^{l_3}-1)\frac{k2^{k+1}}{2^{|N|-1}}$, so adding it to our current case, it leaves
		\begin{align*}
			& \frac{(k-l')2^{k+1-l_2}(2^{k+1-l_3}-1)}{2^{|N|-1-l-l'}} + \alpha \\  > ~ & 
			\frac{k2^{k+1}(2^{k+1}-2)}{2^{|N|-1}} - \frac{k(2^{k+1}-2)}{2^{|N|-1}} + \frac{k(2^{k+2}-2^{l_3 +1})}{2^{|N|-1}} \\
			\ge ~ & \frac{k2^{k+1}(2^{k+1}-2)}{2^{|N|-1}} - \frac{k(2^{k+1}-2)}{2^{|N|-1}} + \frac{k2^{k+1}}{2^{|N|-1}} \\
			> ~ & \frac{k2^{k+1}(2^{k+1}-2)}{2^{|N|-1}} 
		\end{align*}
		so the sum of the two summands (defined by Case~$3$ and Case~$5$) increases in this case.
		
		\item $l_3 > 0 \wedge l_1 > 0$:
		\begin{small}	
			\begin{align*}
				& \frac{(k-l')2^{k+1-l_2}(2^{k+1-l_3}-1) - l_1 {k+1-l_2 \choose \floor{\frac{k+1-l_2}{2}}}{k+1-l_3 \choose \ceil{\frac{k+1-l_3}{2}}}}{2^{|N|-1-l-l'}} \\
				\ge ~ & \frac{(k-l')2^{k+1-l_2}(2^{k+1-l_3}-1) - l_1 2^{k-l_2}2^{k-l_3}}{2^{|N|-1-l-l'}}%
				\\
				= ~ & \frac{1}{2^{|N|-1}}\Big(2^{l_1 + l'}(k-l')2^{k+1}(2^{k+1}-2^{l_3}) - l_1 2^{2k+l_1 +l'}\Big) \\
				= ~ & \frac{1}{2^{|N|-1}}\Big(2^{l_1 + l'}k2^{k+1}(2^{k+1}- 2)- 2^{l_1 + l'}k2^{k+1}(2^{l_3}-2) \\
				& - 2^{l_1 + l'}l'2^{k+1}(2^{k+1}-2^{l_3}) - l_1 2^{2k+l_1 +l'}\Big) \\
				= ~ & \frac{k2^{k+1}(2^{k+1}-2)}{2^{|N|-1}} \\
				& + \frac{1}{2^{|N|-1}}\Big(k2^{2k+2 + l_1 + l'} - k2^{k+2 + l_1 + l'} -k2^{2k+2} %
				\\
				& + k2^{k+2} - k2^{k+1+l_1 + l_3 + l'} + k2^{k+2+l_1 +l'} %
				\\
				& - l'2^{2k+2 + l_1 + l'}  + l'2^{k+1 + l_1 + l_3 + l'} - l_1 2^{2k+l_1 +l'}\Big) \\
				= ~ & \frac{k2^{k+1}(2^{k+1}-2)}{2^{|N|-1}} \\
				& + \frac{1}{2^{|N|-1}}\Big((k-l')2^{2k+2 + l_1 + l'}  -k2^{2k+2} + k2^{k+2} \\
				& - (k-l')2^{k+1+l_1 + l_3 + l'}  - l_1 2^{2k+l_1 +l'}\Big) \\
				\stackrel{(\ref{k:inequality})}{\ge} ~ & \frac{k2^{k+1}(2^{k+1}-2)}{2^{|N|-1}} \\
				& + \frac{1}{2^{|N|-1}}\Big((k-l')2^{2k+1 + l_1 + l'} +k2^{2k+1 + l_1} -k2^{2k+2}  \\
				& + k2^{k+2} - (k-l')2^{k+1+l_1 + l_3 + l'}  - l_1 2^{2k+l_1 +l'}\Big) \\
				\ge ~ & \frac{k2^{k+1}(2^{k+1}-2)}{2^{|N|-1}} \\
				& + \frac{1}{2^{|N|-1}}\Big((k-l')2^{2k + l_1 + l'} + k2^{k+2} \\ 
				& - (k-l')2^{k+1+l_1 + l_3 + l'} \Big) \\
			\end{align*}
		\end{small}
		and for $\alpha=(2^{l_1 + l_3 + l'}-1)\frac{k2^{k+1}}{2^{|N|-1}}$ 	
		(from Case~$5$),
		\begin{align*}
			\alpha + & \frac{(k-l')2^{k+1-l_2}(2^{k+1-l_3}-1) - l_1 {k+1-l_2 \choose \floor{\frac{k+1-l_2}{2}}}{k+1-l_3 \choose \ceil{\frac{k+1-l_3}{2}}}}{2^{|N|-1-l-l'}} %
			\\
			\ge ~ & \frac{k2^{k+1}(2^{k+1}-2)}{2^{|N|-1}}
			+ \frac{1}{2^{|N|-1}}\Big(k2^{2k + l_1}  
			+ k2^{k+2} -k2^{k+1}\Big) \\
			> ~ & \frac{k2^{k+1}(2^{k+1}-2)}{2^{|N|-1}}.
		\end{align*}
		
		\item $l_3 = 0 \wedge l_1 = 0$:
		\begin{align*}
			& \frac{(k-l')2^{k+1-l_2}(2^{k+1}-2)}{2^{|N|-1-l_2-l'}} \\ 
			= ~ & \frac{2^{l_2 +l'}(k-l')2^{k+1-l_2}(2^{k+1}-2)}{2^{|N|-1}} \\
			= ~ & \frac{2^{l'}(k-l')2^{k+1}(2^{k+1}-2)}{2^{|N|-1}} \\ 
			\stackrel{(\ref{k:inequality})}{\ge} ~ & \frac{k2^{k+1}(2^{k+1}-2)}{2^{|N|-1}}.
		\end{align*}
		
		\item $l_3 = 0 \wedge l_1 > 0$:
		\begin{small}
			\begin{align*}
				& \frac{(k-l')2^{k+1-l_2}(2^{k+1}-2) - l_1 {k+1-l_2 \choose \floor{\frac{k+1-l_2}{2}}}{k+1 \choose \floor{\frac{k+1}{2}}}}{2^{|N|-1-l-l'}} %
				\\
				\ge ~ & \frac{1}{2^{|N|-1-l-l'}}\Big((k-l')2^{k+1-l_2}(2^{k+1}-2) - l_1 2^{k-l_2}2^{k}\Big)  \\
				= ~ & \frac{1}{2^{|N|-1}}\Big(2^{l_1 + l'}(k-l')2^{k+1}(2^{k+1}-2) - 2^{l_1 + l'} l_1 2^{2k}\Big)  \\
				\ge ~ & \frac{k2^{k+1}(2^{k+1}-2)}{2^{|N|-1}} \\
				& + \frac{1}{2^{|N|-1}}\Big((2^{l_1 +l'}-1)k2^{k+1}(2^{k+1}-2)\\
				& -2^{l_1 + l'}l' 2^{k+1}(2^{k+1}-2)- l_1 2^{2k+l_1 + l'}\Big)  \\
				= ~ & \frac{k2^{k+1}(2^{k+1}-2)}{2^{|N|-1}} \\
				& + \frac{1}{2^{|N|-1}}\Big((2^{k+1+l_1 +l'}-2^{l_1 +l' +1}-2^{k+1}+2)k2^{k+1} \\
				& -l' 2^{2k+2+l_1 +l'} + l'2^{k+2+l_1 +l'}- l_1 2^{2k+l_1 + l'}\Big)  \\
				= ~ & \frac{k2^{k+1}(2^{k+1}-2)}{2^{|N|-1}} \\
				& + \frac{1}{2^{|N|-1}}\Big(k2^{2k+2+l_1 +l'}-k2^{k+2 +l_1 +l'}-k2^{2k+2} \\
				& +k2^{k+2} -l' 2^{2k+2+l_1 +l'} + l'2^{k+2+l_1 +l'}- l_1 2^{2k+l_1 + l'}\Big)  \\
				= ~ & \frac{k2^{k+1}(2^{k+1}-2)}{2^{|N|-1}} \\
				& + \frac{1}{2^{|N|-1}}\Big((k-l')2^{2k+1+l_1 +l'}+(k-l')2^{2k+l_1 +l'} \\
				& +(k-l')2^{2k+l_1 +l'} - (k-l')2^{k+2 +l_1 +l'}-k2^{2k+2} \\
				& +k2^{k+2} - l_1 2^{2k+l_1 + l'}\Big)  \\
				\ge ~ & \frac{k2^{k+1}(2^{k+1}-2)}{2^{|N|-1}} \\
				& + \frac{1}{2^{|N|-1}}\Big((k-l')2^{2k+1+l_1 +l'}+(k-l')2^{2k+l_1 +l'} \\
				& - (k-l')2^{k+2 +l_1 +l'}-k2^{2k+2}+k2^{k+2} \Big).  
			\end{align*}
		\end{small}
		If $l'=0$, then
		\begin{small}
			\[
			k(2^{2k+1+l_1}+2^{2k+l_1}- 2^{k+2 +l_1}-2^{2k+2}+2^{k+2}) > k2^{k+2} > 0, 
			\]
		\end{small}
		and otherwise,
		\begin{small}
			\begin{align*}
				(k-l') & 2^{2k+1+l_1 +l'} +(k-l')2^{2k+l_1 +l'} - (k-l')2^{k+2 +l_1 +l'}\\
				& -k2^{2k+2}+k2^{k+2}	\\
				\stackrel{(\ref{k:inequality})}{\ge} ~ & (k-l')2^{2k+l_1 +l'}- (k-l')2^{k+2 +l_1 +l'}
				+k2^{k+2} \\
				> ~ & k2^{k+2} \\
				> ~ & 0, 
			\end{align*}
		\end{small}
		so the summand increases also in this case.
	\end{enumerate}	
	Note that if we removed some whole pair $(a_i,b_i)$, $i\in\{1,\ldots,k\}$, from $A$, we would get exactly the same number of coalitions for the summands as in the case of removing only one of the players from the pair, i.e., the new game would just contain less players, therefore, the summand would be even larger (since the denominator would be smaller at that time).
	
	Removing players from the rest of the groups will only increase even more the summands above. Now, we will focus on the coalitions defined by Case~$1$, Case~$2$, and Case~$4$. Moreover, because the increase of the summand considered above depends on the summand defined by Case~$5$, we will not consider it anymore for the three cases.

	First, let us consider removing $l_1 + l_2 + l_3 + l' = l + l' \le k$ players from $T$, $Y'$, $Z'$, and $A$, respectively, with the assumptions that $l'<k$ and no whole pair $(a_i,b_i)$ for any $i\in\{1,\ldots,k\}$ is removed from the game.
	Let us focus on the summands defined by Case~$4$ and by Case~$1$:
	\begin{enumerate}	
		\item For $l_3 > 0 \wedge l_1 = 0$, the summand by Case~$4$ is equal to
		\begin{align*}
			& \frac{(k-l')2^{n-l_2}(2^{k+1-l_3}-1)}{2^{|N|-1-l-l'}} 	\\
			= ~ &  \frac{2^{l'}(k-l')2^{n}(2^{k+1}-2^{l_3})}{2^{|N|-1}} \\
			= ~ &  \frac{2^{l'}(k-l')2^{n}(2^{k+1}-2)-2^{l'}(k-l')2^{n}(2^{l_3}-2)}{2^{|N|-1}} \\
			= ~ &  \frac{k2^n (2^{k+1}-2)}{2^{|N|-1}} 
			\\& + \frac{2^n}{2^{|N|-1}}\Big((2^{l'}-1)k(2^{k+1}-2) -2^{l'}l'(2^{k+1}-2) \\
			& -2^{l'}(k-l')(2^{l_3}-2)\Big) \\
			= ~ & \frac{k2^n (2^{k+1}-2)}{2^{|N|-1}} \\
			&  + \frac{2^n}{2^{|N|-1}}\Big(k2^{k+1+l'}-k2^{k+1}-k2^{l'+1}+2k- l'2^{k+1+l'} \\
			&+l'2^{l'+1}
			-k2^{l_3+l'} +k2^{l'+1}+l'2^{l_3 +l'}-l'2^{l'+1}\Big) \\
			= ~ &  \frac{k2^n (2^{k+1}-2)}{2^{|N|-1}} \\
			& + \frac{2^n}{2^{|N|-1}}\Big(k2^{k+1+l'}-k2^{k+1}+2k - l'2^{k+1+l'} \\
			& -k2^{l_3+l'} +l'2^{l_3 +l'}\Big) \\
			= ~ &  \frac{k2^n (2^{k+1}-2)}{2^{|N|-1}} \\
			&  + \frac{2^n}{2^{|N|-1}}\Big(2^{k+1}(2^{l'}(k-l')-k)+2k -(k-l')2^{l_3+l'}\Big) \\
			\stackrel{(\ref{k:inequality})}{\ge}  &  \frac{k2^n (2^{k+1}-2)}{2^{|N|-1}}  + \frac{2^n}{2^{|N|-1}}\Big(2k -(k-l')2^{l_3+l'}\Big).
		\end{align*}
		The new summand defined by Case~$6$ has increased by at least $\gamma = (2^{l_3 + l'}-1)\frac{k2^n}{2^{|N|-1}}$, so
		\begin{small}
			\begin{align*}
				& \frac{(k-l')2^{n-l_2}(2^{k+1-l_3}-1)}{2^{|N|-1-l-l'}}  + \gamma	\\
				\ge ~ &  \frac{k2^n (2^{k+1}-2)}{2^{|N|-1}}  + \frac{2^n}{2^{|N|-1}}\Big(2k -(k-l')2^{l_3+l'}\Big) \\
				&  +  (2^{l_3 + l'}-1)\frac{k2^n}{2^{|N|-1}} \\
				= ~ &  \frac{k2^n (2^{k+1}-2)}{2^{|N|-1}}  + \frac{2^n}{2^{|N|-1}}\Big(k -(k-l')2^{l_3+l'} + k2^{l_3+l'}\Big)  \\
				\ge ~ &  \frac{k2^n (2^{k+1}-2)}{2^{|N|-1}}  + \frac{k2^n}{2^{|N|-1}} 
			\end{align*}
		\end{small}
		therefore the sum of the two summands is larger than the sum of these summands in the old game for $l'>0$, 
		as well as for $l'=0$, since
		\begin{small}
			\begin{align*}
				& \frac{k2^{n-l_2}(2^{k+1-l_3}-1)}{2^{|N|-1-l}}  + \gamma	\\
				= ~ &  \frac{k2^n (2^{k+1}-2)-k2^n(2^{l_3}-2)}{2^{|N|-1}}   +  (2^{l_3}-1)\frac{k2^n}{2^{|N|-1}} \\
				= ~ &  \frac{k2^n (2^{k+1}-2)}{2^{|N|-1}}  + \frac{k2^n}{2^{|N|-1}}\Big(2^{l_3}-1-2^{l_3}+2\Big)  \\
				= ~ &  \frac{k2^n (2^{k+1}-2)}{2^{|N|-1}}  + \frac{k2^n}{2^{|N|-1}} 
			\end{align*}
		\end{small}
		\noindent
		At the same time (assuming we do not delete other players forming these coalitions), the summand by Case~$1$ is at least
		\begin{align*}
			& \frac{2k(2^k -1)2^{k-l'}(2^{n-k-1}+1)}{2^{|N|-1-l-l'}} \\ 
			= ~ & \frac{2^{l}2k(2^k -1)(2^{n-1}+2^{k})}{2^{|N|-1}}  \\
			= ~ & \frac{2^{l-1}2k(2^k -1)(2^{n}+2^{k+1})}{2^{|N|-1}} \\	
			\ge ~ & \frac{2k(2^k -1)(2^{n}+2^{k+1})}{2^{|N|-1}}  \\
			\stackrel{(\ref{app:xi})}{>} ~ & \frac{2k(2^k -1)\xi}{2^{|N|-1}}.
		\end{align*}
		
		\item For $l_3>0 \wedge l_1 > 0$, the summand by Case~$4$ is at least
		\begin{align*}
			& \frac{(k-l')2^{n-l_2}(2^{k+1-l_3}-1)- l_1 {n-l_2 \choose \floor{\frac{n-l_2}{2}}}{k+1-l_3 \choose \floor{\frac{k+1-l_3}{2}}}}{2^{|N|-1-l-l'}}	\\
			\ge ~ & \frac{k2^n (2^{k+1}-2)}{2^{|N|-1}} 
			\\& + \frac{1}{2^{|N|-1}}\Big((2^{l_1+l'}-1)2^n k(2^{k+1}-2)\\
			&  -2^{l_1+l'}l'2^n (2^{k+1}-2) \\
			& -2^{l_1+l'}(k-l')2^n(2^{l_3}-2)
			-2^{l + l'}l_1 2^{n-1 - l_2} 2^{k-l_3}\Big) \\
			= ~ & \frac{k2^n (2^{k+1}-2)}{2^{|N|-1}} 
			\\& + \frac{2^n}{2^{|N|-1}}\Big(k2^{k+l_1+l'+1}-k2^{l_1 + l' + 1}-k2^{k+1}+2k   \\
			& -l'2^{k+l_1+l'+1}+l'2^{l_1 + l' +1}-2^{l_1+l_3+l'}(k-l') \\
			& +k2^{l_1 + l' +1}-l'2^{l_1 + l' +1} -l_1 2^{k+l_1 +l' -1}\Big) \\
			= ~ & \frac{k2^n (2^{k+1}-2)}{2^{|N|-1}} 
			\\& + \frac{2^n}{2^{|N|-1}}\Big(4(k-l')2^{k+l_1+l'-1}\\
			& -k2^{k+1}+2k -2^{l_1+l_3+l'}(k-l')
			-l_1 2^{k+l_1 +l' -1}\Big) \\
			\ge ~ & \frac{k2^n (2^{k+1}-2)}{2^{|N|-1}} 
			\\& + \frac{2^n}{2^{|N|-1}}\Big((k-l')2^{k+l_1+l'}-k2^{k+1}+2k \Big) \\
			\stackrel{(\ref{k:inequality})}{\ge} ~ &  \frac{k2^n (2^{k+1}-2)}{2^{|N|-1}} + \frac{k2^{n+1}}{2^{|N|-1}} \\
		\end{align*}
		so the summand nondecreases also in this case. The summand by Case~$1$ is exactly as in the previous situation.
		
		\item Let $l_3=0 \wedge l_1 = 0$.  The summand of Case~$1$ is at least
		\begin{small}
			\begin{align*}
				\frac{2k(2^k -1)2^{k-l'}(2^{n-k-1}+1)}{2^{|N|-1-l_2-l'}} & = ~  \frac{2^{l_2}2k(2^k -1)(2^{n-1}+2^{k})}{2^{|N|-1}}.
			\end{align*}
		\end{small}
		If $l_2>0$, then
		\begin{small}
			\begin{align*}
				\frac{2^{l_2}2k(2^k -1)(2^{n-1}+2^{k})}{2^{|N|-1}} \ge ~ & \frac{2^{l_2 -1}2k(2^k -1)(2^{n}+2^{k+1})}{2^{|N|-1}}	\\
				\stackrel{(\ref{app:xi})}{>} ~ &  \frac{2k(2^k -1)\xi}{2^{|N|-1}},
			\end{align*}
		\end{small}
		and otherwise (and let $l'>0$ because we would make no changes among the players forming the coalitions counted in the cases), 
		it is at least
		\begin{align*}
			& 	\frac{2k(2^k -1)2^{k-l'}(2^{n-k-1}+1)}{2^{|N|-1-l'}} \\ = ~ & \frac{2k(2^k -1)(2^{n-1}+2^{k})}{2^{|N|-1}} \\
			= ~ & \frac{2k(2^k -1)2^{n}}{2^{|N|-1}} -  \frac{2k(2^k -1)(2^{n-1}-2^{k})}{2^{|N|-1}} \\
			\stackrel{(\ref{app:xi})}{\ge} ~ & \frac{2k(2^k -1)\xi}{2^{|N|-1}} -  
			\frac{k(2^k -1)(2^{n}-2^{k+1})}{2^{|N|-1}}.
		\end{align*}
		The summand of Case~$4$ 
		for $l'>0$
		is equal to
		\begin{align*}
			& \frac{(k-l')2^{n-l_2}(2^{k+1}-2)}{2^{|N|-1-l_2-l'}}\\  = ~ & \frac{2^{l'}(k-l')2^{n}(2^{k+1}-2)}{2^{|N|-1}} \\
			= ~ &  \frac{k 2^n (2^{k+1}-2)}{2^{|N|-1}} \\
			& + \frac{(2^{l'}-1)2k2^{n}(2^{k}-1)-2^{l'}l'2^{n}(2^{k+1}-2)}{2^{|N|-1}} \\
			= ~ &  \frac{k 2^n (2^{k+1}-2)}{2^{|N|-1}}  + \frac{2k2^{n-1}(2^{k}-1)}{2^{|N|-1}}  \\	
			&+ \frac{(2^{l'}-1)2k2^{n-1}(2^{k}-1)}{2^{|N|-1}} \\
			&+ \frac{(2^{l'}-2)2k2^{n-1}(2^{k}-1) -2^{l'}l'2^{n}(2^{k+1}-2)}{2^{|N|-1}} \\
			= ~ & \frac{k 2^n (2^{k+1}-2)}{2^{|N|-1}}  + \frac{2k2^{n-1}(2^{k}-1)}{2^{|N|-1}}  \\
			&+ 2^{n-1} (2^{k+1}-2) \frac{2^{l'+1}k -3k -2^{l'+1}l'}{2^{|N|-1}} \\
			= ~ & \frac{k 2^n (2^{k+1}-2)}{2^{|N|-1}}  + \frac{2k2^{n-1}(2^{k}-1)}{2^{|N|-1}}  \\
			&+ 2^{n-1} (2^{k+1}-2) \frac{(2^{l'+1}-3)k -2^{l'+1}l'}{2^{|N|-1}}
		\end{align*}
		and since $k\ge 4$ and $l'\in\{1,\ldots,k-1\}$,
		\begin{align*}
			(2^{l'+1}-3)k -2^{l'+1}l' & \ge (2^{l'+1}-3)(l'+1) -2^{l'+1}l' \\
			& = 2^{l'+1} - 3l'-3\ge 0
		\end{align*}
		for $l'\ge 3$, for $l'=1$,
		\[
		(2^{l'+1}-3)k -2^{l'+1}l' = k-4 \ge 0,
		\]
		and finally, for $l'=2$,
		\[
		(2^{l'+1}-3)k -2^{l'+1}l' = 5k-16 \ge 20-16 > 0.
		\]
		Therefore,
		\begin{small}
			\[
			\frac{(k-l')2^{n-l_2}(2^{k+1}-2)}{2^{|N|-1-l_2-l'}} \ge  \frac{k 2^n (2^{k+1}-2)}{2^{|N|-1}}  + \frac{k2^{n}(2^{k}-1)}{2^{|N|-1}},
			\]
		\end{small}
		which together with the summand by Case~$1$, is greater than 	
		the both summands together in the old game. 	
		If $l'=0$, the summand by Case~$4$ stays unchanged.
		
		\item For $l_3=0 \wedge l_1 > 0$, the summand of Case~$1$ is as in the first two situations. The summand defined by Case~$4$ in the new game is equal to
		\begin{align*}
			& \frac{(k-l')2^{n-l_2}(2^{k+1}-2)-l_1 {n-l_2 \choose \floor{\frac{n-l_2}{2}}}{k+1 \choose \floor{\frac{k+1}{2}}}}{2^{|N|-1-l-l'}}	\\
			\ge ~ & \frac{k2^n (2^{k+1}-2)}{2^{|N|-1}} 
			\\& + \frac{1}{2^{|N|-1}}\Big((2^{l_1+l'}-1)2^n k(2^{k+1}-2) \\
			& -2^{l_1+l'}l'2^n (2^{k+1}-2) -2^{l_1 + l'}l_1 2^{n-1+k}\Big) \\
			= ~ &  \frac{k2^n (2^{k+1}-2)}{2^{|N|-1}} 
			\\& + \frac{2^n}{2^{|N|-1}}\Big(k2^{k+l_1+l'+1}-k2^{l_1 + l' + 1}-k2^{k+1}+2k \\
			& -l'2^{k+l_1+l'+1} +l'2^{l_1 + l' +1} -l_1 2^{k+l_1 +l' -1}\Big) \\
			= ~ &  \frac{k2^n (2^{k+1}-2)}{2^{|N|-1}} 
			\\& + \frac{2^n}{2^{|N|-1}}\Big(4(k-l')2^{k+l_1+l'-1}-(k-l')2^{l_1 + l' + 1} \\ 
			& -k2^{k+1}+2k -l_1 2^{k+l_1 +l' -1}\Big) \\
			\ge ~ &  \frac{k2^n (2^{k+1}-2)}{2^{|N|-1}} \\
			& + \frac{2^n}{2^{|N|-1}}\Big((k-l')2^{k+l_1+l'}-k2^{k+1}+2k \Big) \\
			\stackrel{(\ref{k:inequality})}{\ge} ~ &  \frac{k2^n (2^{k+1}-2)}{2^{|N|-1}} 
			+ \frac{2^n}{2^{|N|-1}}\Big(k2^{k+l_1}-k2^{k+1}+2k \Big) \\
			> ~ &  \frac{k2^n (2^{k+1}-2)}{2^{|N|-1}}, 
		\end{align*}
		so the summand increases.
		The summand by Case~$1$ reacts as it did in the first two cases.	
	\end{enumerate}  
	The summands defined by Case~$1$ and Case~$2$ may be able to be decreased to $0$ by deleting one player from 
	one player from $C$ and one player from $E$, while decreasing the denominator only by a factor of $2^2 = 4$ (or by $2$ if we delete only one of the players, which also will not change the numerator of the other summand). 
	If we delete at least one more player from $A$ from a pair $(a_i, b_i)$ from which we have already removed a player in the cases above or any other player that forms some coalitions from Case~$1$ also causing the decrease of the summand to $0$, it will not change the number of coalitions counted in the other summand, but will decrease its denominator and therefore, by (\ref{app:thelargest:1or2}), the sum will still increase. We get the same situation, if we also delete a player or players decreasing the summand defined by Case~$2$ to $0$.    
	
	Summarizing the whole analysis above:
	\begin{enumerate}
		\item Removing $k$ players as it was in the previous implication but for any set corresponding to a truth assignment for the first $k$ variables, the index increases in the new game created from $\mathcal{G}$ and it nondecrease in the game created from $\mathcal{G}'$.
		\item If we remove some of the players forming coalitions from Case~$5$ and Case~$6$, there are two possibilities, i.e.,
		\begin{itemize}
			\item Either we remove $k$ only from $Z^{*}$ excluding all the coalitions from the index in the worst situation, but by (\ref{app:thelargest:5and6}), the new index is greater due to the coalitions from Case~$4$, or
			\item We remove less of the players---deleting or not other players forming coalitions from the other cases---and then, the summands defined by Case~$5$ and by Case~$6$ 
			together increase (note that the sets of players that define the both cases are disjoint
			except for $Z^*$), while each separately nondecreases---however, without removing any other player changing some other summands, the other summands increase. 
		\end{itemize}
		\item For $\mathcal{G}$, removing players forming the coalitions from Case~$3$ makes the sum of this summand and the summand defined by Case~$5$ in the new game greater.
		\item  If we remove some players forming the coalitions from Case~$4$, this summand (eventuelly together with the summand by Case~$6$) increases, and
		even 
		when the summand by Case~$1$ decreases---assuming we remove only the players that are part of the coalitions by Case~$4$---the sum of the two summands still increases. 
		\item If we remove the players that are crucial for the summands by Case~$1$ and Case~$2$ (not studied in the previous situation) and they decrease even to $0$, the index increases due to (\ref{app:thelargest:1or2})	
		and the fact that these players are not a part of any coalition in all the other cases 
		(anymore).
	\end{enumerate}

	The indices that player~$1$ has after removing other combinations of players are not smaller than the cases we presented above. Thus the Penrose--Banzhaf index of player~$1$
	cannot decrease for $\mathcal{G}$ and nonincrease for $\mathcal{G}'$ by deleting up to $k$ players from $N\setminus\{1\}$ in
	game~$\mathcal{G}$ and from $N'\setminus\{1\}$ in game~$\mathcal{G}'$, so we have no-instances of our control
	problems.
\end{proofs}

\section{Full Proof of the Second Part of Theorem~\ref{thm:deleting-remaining-PBI}}

We start by defining a problem introduced by Kaczmarek and Rothe~\cite{kac-rot:c:control-by-adding-players-to-change-or-maintain-shapley-shubik-or-the-penrose-banzhaf-power-index-in-WVGs-is-complete-for-np-pp,kac-rot:t:control-by-adding-players-to-change-or-maintain-shapley-shubik-or-the-penrose-banzhaf-power-index-in-WVGs-is-complete-for-np-pp}:
\EP{\textsc{E-Exact-SAT}}
{A boolean formula $\phi$ with $n$ variables $x_1,\ldots,x_n$, an
	integer~$k$, $1\le k\le n$, and 
	an integer~$\ell>0$.}
{Is there an assignment to the first $k$ variables $x_1,\ldots,x_k$
	such that \emph{exactly} $\ell$ assignments to the remaining
	$n-k$ variables $x_{k+1},\ldots,x_n$ 
	satisfy~$\phi$?}

We now show that in the definition of \textsc{E-Exact-SAT}, we can assume that $\ell$ is not a power of~$2$, i.e., $\ell$ is can be written as
\[
\ell = 2^{\delta_1} + \cdots + 2^{\delta_h}
\]
for $h, \delta_1,\ldots,\delta_h\in\mathbb{N}$
with $h\ge 2$
and $\delta_1>\cdots>\delta_h$. Let us call this special case \textsc{E-Exact-SAT$^{*}$}.
We now show that this special case is no easier to solve than the general problem.

\begin{lemma}
	\textsc{E-Exact-SAT} can be reduced to \textsc{E-Exact-SAT}$^{*}$.
\end{lemma}

\begin{proof}
	We provide the following reduction from \textsc{E-Exact-SAT}: For an instance $(\phi, k, \ell)$, we define
	\[
	\phi' = \phi \wedge (x_{n+1}\vee x_{n+2})
	\]
	for two new variables $x_{n+1}$ and $x_{n+2}$ (added to the end of the variable list for $(\phi, k, \ell)$).
	The last clause is always true for three assignments for the two variables (independently from the assignment for all the other variables): $(1,1)$, $(1,0)$, and $(0,1)$. The instance for the special case is defined as $(\phi',k,3\ell)$.
	
	First, let 
	\[
	\ell = 2^{\delta_1} + \cdots + 2^{\delta_h}
	\]  
	for any $h\in\{1,\ldots,n\}$. Then
	\begin{align*}
		3\ell & = 3\cdot(2^{\delta_1} + \cdots + 2^{\delta_h}) \\
		& = 2^{\delta_1 +1} + \cdots + 2^{\delta_h +1} + 2^{\delta_1} + \cdots + 2^{\delta_h} \\
		& = \ell' +  2^{\delta_h}
	\end{align*}
	for some $\ell' > 2^{\delta_h}$, i.e., $3\ell$ needs more than one $1$ in its binary representation. So, it is a valid value for an input for \textsc{E-Exact-SAT$^{*}$}.
	
	\proofonlyif
	Let $(\phi, k, \ell)\in\textsc{E-Exact-SAT}$. The same assignment for the first $k$ variables being the answer for this input is an answer for the input created in our reduction. Let us fix this assignment. Then exactly $\ell$ assignments for the rest variables in $\phi$ satisfy the formula, and exactly $3\ell$ assignments (extended for the new two variables) satisfy~$\phi'$.
	
	\proofif
	Let $(\phi', k, 3\ell)\in $ \textsc{E-Exact-SAT$^{*}$}. Let us fix again the first $k$ assignment being an answer for the problem. Since $x_{n+1}$ and $x_{n+2}$ appear only in the last clause, which is not a part of $\phi$, for each of 
	the assignments for them that make the clause true, there are exactly the same assignments for the other $n-k$ variables that satisfy together with the $k+2$ fixed values formula~$\phi'$. That means, for each of these assignments for $x_{n+1}$ and $x_{n+2}$, there are exactly $\ell$ assignments, and they satisfy---together with the fixed $k$ first values---also formula~$\phi$.
\end{proof}

\begin{proofs} [Full proof of $\np^{\PP}$-hardness of \textsc{Control-by-Deleting-Players-to-Maintain-PBI} (second part of Theorem~\ref{thm:deleting-remaining-PBI}).]
	
	We will prove $\np^{\PP}$-hardness using a reduction from the problem
	\textsc{E-Exact-SAT$^{*}$} defined above.
	Let $(\phi, k, \ell)$ be a given instance of
	this problem, where $\phi$ is a Boolean formula in CNF
	with variables $x_1,\dots,x_n$ and $m$ clauses, and $4\le k<n$. 
	Also, $\ell$ and $l$ will represent two different notations. 
	Before we construct
	instances of our control problems from $(\phi, k, \ell)$, we need to
	choose some numbers and introduce some notations.
	
	Let $\delta_1 > \cdots >\delta_h$ be the values fulfilling
	\[
	\ell = 2^{\delta_1} + \cdots + 2^{\delta_h}
	\]
	for $h, \delta_1\le n$ and $h\ge2$. 
	
	Let $z^{*}_1 = (k+2)z'$ for $z'$ defined in Table~\ref{tab:deleting-maintain-PBI} and for $i\in\{2,\ldots,k\}$
	\[
	z^{*}_i = z^{*}_1 + \sum_{j=1}^{i-1}z^{*}_j.
	\]
	Let 
	$t \in \mathbb{N}$ be such that 
	\begin{equation}\label{m:def:t}
		10^{t} > \max\left\{2^{\ceil{\log_{2} n}+1}, 2z^{*}_k\right\}.
	\end{equation}
	For that $t$, let $A$, $B$, $C$, $a_i$, $b_i$, and $c_{j,s}$ for $i \in \{1,\ldots,n\}$, $j\in\{1,\ldots,m\}$, and $s\in\{0,\ldots,r\}$, and $q'$ be defined as in Definition~\ref{def:prereduction}.
	Finally, let
	\begin{eqnarray*}
		a_i '   & = & a_i \cdot 10^{t(m+1)+n},\\
		b_i '   & = & b_i \cdot 10^{t(m+1)+n}, \text{ and }\\
		c_{j,s}' & = & c_{j,s} \cdot 10^{t(m+1)+n},
	\end{eqnarray*}
	which together form the weight vector~$W_E$.
	
	Now, we construct
	the instances of our control problems.  Let $k$ be the limit for the
	number of players that can be deleted.
	Further, let 
	\[
	q'' = 10^{t(m+1)+n}q'.
	\]
	Note that,
	as shown by Kaczmarek and Rothe~\cite{kac-rot:c:control-by-adding-players-to-change-or-maintain-shapley-shubik-or-the-penrose-banzhaf-power-index-in-WVGs-is-complete-for-np-pp,kac-rot:t:control-by-adding-players-to-change-or-maintain-shapley-shubik-or-the-penrose-banzhaf-power-index-in-WVGs-is-complete-for-np-pp}, 
	the first summand of $q'$ can be achieved only by
	taking exactly one element from $\{a_i,b_i\}$ for all
	$i\in\{1,\ldots,n\}$ and summing them up, and that
	the second summand of $q'$ cannot be achieved only by the elements
	$c_{j,s}$, with $j\in\{1,\ldots,m\}$ and $s\in\{0,\ldots,r\}$ (in the case of $q''$, the situation is analogous).
	Let $N$ be the set of players in the game $\mathcal{G}$ divided into the groups with the weights presented in Table~\ref{tab:deleting-maintain-PBI}. Finally, define the quota of the weighted voting game by
	\begin{equation}\label{m:def:q}
		q=2\cdot \left(w_A + w_B  + w_C + w_E + 10^{t}\right) +1.
	\end{equation}
	
	\begin{table*}\setcounter{table}{1}
		\begin{center}
			\caption{\label{tab:deleting-maintain-PBI}
				Groups of players in the proof of the second part of
				Theorem~\ref{thm:deleting-remaining-PBI}, with their 
				numbers and weights (with $i\in\{1,\ldots,h\}$)
			}
			\renewcommand{\arraystretch}{1.5}
			\begin{tabular}{c|c|c}
				\toprule
				\textbf{Group} & \textbf{Number of Players} & \textbf{Weights} \\
				\midrule
				distinguished player~$p$ & $1$ & $1$ \\
				\midrule
				$A$ & $2k$ & $W_A$ \\
				\midrule
				$B$ & $2n-2k$ & $W_B$ \\
				\midrule
				$C$ & $m(r+1)$ & $W_C$ \\
				\midrule
				$D$ & $k$ & \makecell{$q-q'-i' x-x'$ \\ for $i' \in\{0,\ldots,k-1\}$} \\
				\midrule
				$E$ & $2n+m(r+1)$ & $W_E$ \\
				\midrule
				$F$ & $1$ & $q-q''-x'$ \\
				\midrule
				$L_i$ & $\delta_i$ & $\delta_i ' = (\delta_i +1)\delta_{i-1}'$ with $\delta_1 ' = 1$ \\
				\midrule
				$S_i$ & $k^2 (k+2) (\delta_i +1)$ & \makecell{$q-(a_{i'} + b_{i'})-jy-j' z - j'' \delta_i '$ \\ for $i' \in \{1,\ldots,k\}$, $j\in\{0,\ldots,k+1\}$, \\  
					$j'\in\{0,\ldots,k-1\}$, 
					and $j''\in\{0,\ldots,\delta_i\}$} \\
				\midrule
				$T_i$ & $k^2 (n+3) (\delta_i+1)$ & \makecell{$q-(a_{i'} + b_{i'})-jy'-j' z' - j'' \delta_i '$ \\ for $i' \in \{1,\ldots,k\}$, $j\in\{0,\ldots,n+2\}$, \\ 
					$j'\in\{0,\ldots,k-1\}$, 
					and $j''\in\{0,\ldots,\delta_i\}$} \\
				\midrule
				$U_i$ & $k^2 (\delta_i+1)$ & \makecell{$q-(a_{i'} + b_{i'})-jy^{*} - j' \delta_i '$ \\ for $i'\in\{1,\ldots,k\}$, $j\in\{1,\ldots,k\}$, \\ and $j'\in\{0,\ldots,\delta_i\}$} \\
				\midrule
				$V_i$ & $k (n+5) (\delta_i+1)$ & \makecell{$q-jy^{**}-z^{*}_{i'} - j' \delta_i '$ \\ for $i'\in\{1,\ldots,k\}$, $j\in\{0,0,1,\ldots,n+2,n+2\}$, \\ and $j'\in\{0,\ldots,\delta_i\}$} \\
				\midrule
				$X$ & $k$ & $x = (\delta_h +1) \delta_h '$ \\
				\midrule
				$X'$ & $2k$ & $x' = (k+1)x$ \\
				\midrule
				$Y$ & $k+1$ & $y = (2k+1)x'$ \\
				\midrule
				$Y'$ & $n+2$ & $y' = (k+2)y$ \\
				\midrule
				$Y^{*}$ & $k+1$ & $y^{*} = (n+3)y'$ \\
				\midrule
				$Y^{**}$ & $n+2$ & $y^{**} = (k+2)y^{*}$ \\
				\midrule
				$Z$ & $k$ & $z = (n+3)y^{**}$ \\
				\midrule
				$Z'$ & $k$ & $z' = (k+1)z$ \\
				\midrule
				$Z^{*}$ & $k$ & $z^{*}_1,\ldots,z^{*}_k$ \\
				\bottomrule
			\end{tabular}
		\end{center}
	\end{table*}
	
	Let us first discuss which coalitions player~$1$ can be pivotal for
	in game $\mathcal{G}$.\footnote{For any $M\subseteq N\setminus\{1\}$, $1$ is pivotal for some subset of the coalitions in $\mathcal{G}_{\setminus M}$.}
	Let $S = \bigcup_{i\in\{1,\ldots,h\}} S_i$, $T = \bigcup_{i\in\{1,\ldots,h\}} T_i$, $U = \bigcup_{i\in\{1,\ldots,h\}} U_i$, $V = \bigcup_{i\in\{1,\ldots,h\}} V_i$, and $L = \bigcup_{i\in\{1,\ldots,h\}} L_i$.
	Player~$1$ is pivotal for those coalitions of players in
	$N\setminus \{1\}$ whose total weight is $q-1$.
	First, note that any two players from $D\cup 
	F \cup S \cup T \cup U \cup V$
	together have a
	weight larger than $q$.  Therefore, at most one player from $D\cup 
	F\cup S \cup T \cup U \cup V$
	can be in any coalition player~$1$ is pivotal for.  Moreover,
	by~(\ref{m:def:q}), all other players
	together have a total weight smaller than $q-1$. 
	That means that any coalition $K \subseteq N\setminus \{1\}$ with a total weight of $q-1$ has to contain \emph{exactly} one of
	the players in $D\cup 
	F \cup  S \cup T \cup U \cup V$.  Now, whether this player is in $D$, $F$, $S$,
	$T$, 
	$U$, or $V$
	has consequences as to which other players will also be in such a
	weight-$(q-1)$ coalition~$K$:
	For $i\in\{1,\ldots,h\}$, 	

	\begin{description}\label{app:cases}
		\item[Case 1:]\label{app:case-1} If $K$ contains a player from~$D$ (with weight, say,
		\[
		q-q'-i'x-x'
		\]
		for some~$i'$, $0 \leq i' \leq k-1$), $K$ also has to
		contain those players from $A\cup B \cup 
		C$ whose weights sum up to $q'$, $i'$ players from $X$ with 
		weight~$x$, and a player from $X'$.  
		Also, recall that $q'$ can be achieved only by a set
		of players whose weights take exactly one of the values from
		$\{a_{i''},b_{i''}\}$ for each $i''\in\{1,\ldots,n\}$,
		so $K$ must contain exactly $n$ players from  
		$A \cup B$.
		
		\item[Case 2:]\label{app:case-2} If $K$ contains the player from 
		$F$, $K$ has to contain those players from $E$ whose weights sum up to $q''$ and a player from $X'$.
		
		\item[Case 3:]\label{app:case-3} If $K$ contains a player from $S_i$	
		 (with weight, say,
		\[
		q-(a_{i'} + b_{i'})-jy-j' z - j'' \delta_i '
		\] 
		for some~$i'$, $1 \leq i' \leq k$, some~$j$, $0 \le j \le k+1$, 		
		some~$j'$, 
		$0\le j' \le k-1$, and some~$j''$, $0\le j'' \le \delta_i$),  $K$ also contains some pair of players from $A$, $j$ players from $Y$, $j'$ players from $Z$, and $j''$ players from $L_i$.
		
		\item[Case 4:]\label{app:case-4} If $K$ contains a player from $T_i$ (with weight, say,
		\[
		q-(a_{i'} + b_{i'})-jy'-j' z' - j'' \delta_i '
		\]
		for some~$i'$, $1 \leq i' \leq k$, some~$j$, $0 \le j \le n+2$, 
		some~$j'$, $0\le j' \le k-1$, 
		and some~$j''$, $0\le j'' \le \delta_i$), $K$ also contains some pair of players from $A$, $j$ players from $Y'$, $j'$ players from $Z'$, and $j''$ players from $L_i$.
		
		\item[Case 5:]\label{app:case-5} If $K$ contains a player from $U_i$ (with weight, say,
		\[
		q-(a_{i'} + b_{i'})-jy^{*} - j' \delta_i '
		\]
		for some~$i'$, $1 \leq i' \leq k$, 	
		some~$j$, 
		$1 \le j \le k$, and some~$j'$, $0\le j' \le \delta_i$, $K$ also contains some pair of players from $A$, $j$ players from $Y^{*}$,  a player from $Z^{*}$, and $j'$ players from $L_i$.
		
		\item[Case 6:]\label{app:case-6} If $K$ contains a player from $V_i$ (with weight, say,
		\[
		q-jy^{**}-z_{i'}^{*} - j' \delta_i '
		\]
		for some~$i'$, $1 \leq i' \leq k$, some~$j$, $0 \le j \le n+2$, and some~$j'$, $0\le j' \le \delta_i$, $K$ also contains $j$ players from $Y^{**}$, a player from $Z^{*}$, and $j'$ players from $L_i$.
	\end{description}

	For any $\np$ problem $L$ (represented by an NP machine~$M_L$) and input~$x$, let $\#L(x)$ denote the number of solutions of $M_L$ on input~$x$.
	For example, using the ``standard'' NP machine~$M_{\textsc{SAT}}$ that, given a boolean formula, guesses its truth assignments and verifies whether they are satisfying, $\#\textsc{SAT}(\phi)$ denotes the number of truth assignments satisfying~$\phi$.
	Let $\xi = \#\textsc{SAT}(\phi)$.

	Let \textsc{SubsetSum} (see, e.g.,~\cite{gar-joh:b:int}) be the well-known, $\np$-complete problem: Given positive integer sizes $s_1, \ldots , s_d$ and a positive integer~$\alpha$, is there a subset $H \subseteq \{1, \ldots , d\}$ such that $\sum_{h \in H} s_h = \alpha$?
	Using the result by Kaczmarek and Rothe~\cite{kac-rot:c:control-by-adding-players-to-change-or-maintain-shapley-shubik-or-the-penrose-banzhaf-power-index-in-WVGs-is-complete-for-np-pp,kac-rot:t:control-by-adding-players-to-change-or-maintain-shapley-shubik-or-the-penrose-banzhaf-power-index-in-WVGs-is-complete-for-np-pp},
	\begin{eqnarray*}
		\xi & = & \#\textsc{SubsetSum}(A\cup B\cup C, q')\\
		& = & \#\textsc{SubsetSum}(E,q'').
	\end{eqnarray*}   
	
	Considering all the cases above, player~$1$'s Penrose--Banzhaf indices in game $\mathcal{G}$  is as follows:
	\begin{small}
		\begin{align*}
			& \PenroseBanzhaf(\mathcal{G},1) \\
			& =\frac{2k(2^k -1)\xi+ 2k\xi 
				+ k2^{k+1} (2^{k} - 1)(2^{\delta_1}+\cdots+2^{\delta_h})}{2^{|N|-1}} \\
			& ~~ + \frac{ k2^{n+2} (2^{k} -1)(2^{\delta_1}+\cdots+2^{\delta_h})}{2^{|N|-1}} \\
			& ~~ + \frac{k(2^{k+1}-2)(2^{\delta_1}+\cdots+2^{\delta_h})}{2^{|N|-1}} \\
			& ~~ + \frac{k(2^{n+2}+2)(2^{\delta_1}+\cdots+2^{\delta_h})}{2^{|N|-1}} \\
			& = \frac{2k(2^k -1)\xi+ 2k\xi + k\ell 2^{k+1} (2^{k} - 1)+ k\ell 2^{n+2} (2^{k} -1)}{2^{|N|-1}} \\
			& ~~ +\frac{k\ell (2^{k+1}-2) + k\ell (2^{n+2}+2)}{2^{|N|-1}}. \\
			& =\frac{2k 2^k\xi + k\ell (2^{n+2} + 2^{k+1})2^{k}}{2^{|N|-1}}. \\
		\end{align*}
	\end{small}

	We now show that 
	$(\phi, k, \ell)$ is a
	yes-instance of \textsc{E-Exact-SAT$^{*}$} if and only if $(\mathcal{G}, 1, k)$ as defined above is a yes-instance of
	\textsc{Control-by-Deleting-Players-to-Maintain-PBI}.

	\proofonlyif
	Suppose that $(\phi, k, \ell)$ is a yes-instance of \textsc{E-Exact-SAT$^{*}$},
	i.e., there exists an assignment of $x_1,\ldots,x_k$ such that exactly $\ell$ of assignments of the remaining $n-k$ variables 
	satisfy
	the Boolean formula~$\phi$.  Let us fix one of these
	satisfying assignments.  From this fixed assignment, the first
	$k$ positions correspond to the players $A' \subseteq A$, $|A'|=|A\setminus A'|=k$,
	for which the players $A\setminus A'$ are removed from game~$\mathcal{G}$
	(see~\cite{kac-rot:c:control-by-adding-players-to-change-or-maintain-shapley-shubik-or-the-penrose-banzhaf-power-index-in-WVGs-is-complete-for-np-pp,kac-rot:t:control-by-adding-players-to-change-or-maintain-shapley-shubik-or-the-penrose-banzhaf-power-index-in-WVGs-is-complete-for-np-pp} for details).
	
	Since there are exactly $\ell$ assignments for
	$x_{n-k},\ldots,x_n$ which---together with the fixed assignments for
	$x_1,\ldots,x_k$---satisfy~$\phi$, there are exactly
	$\ell$ subsets of $A\cup B \cup C$ such that the players' weights
	in each subset sum up to~$q'$. Player~$1$  is still pivotal for all coalitions described in Case~$2$ and Case~$6$, for $2k(2^k -1)\ell$ coalitions described in Case~$1$, and is not pivotal for any coalition described in Case~$3$, Case~$4$, and Case~$5$ anymore. Therefore,
	\begin{align*}
		\PenroseBanzhaf( & \mathcal{G}_{\setminus (A\setminus A')},1) \\ 
		& = \frac{2k(2^k -1)\ell + 2k\xi +  k\ell (2^{n+2}+2)}{2^{|N|-1-k}} \\
		& =\frac{2k \xi 2^k + k\ell (2^{k+1} -2) 2^{k} + k\ell (2^{n+2}+2)2^k}{2^{|N|-1}} \\
		& = \frac{2k 2^k \xi + k\ell (2^{n+2}+2^{k+1})2^{k}}{2^{|N|-1}} \\
		& = \PenroseBanzhaf(\mathcal{G},1),
	\end{align*}
	so player~$1$'s Penrose--Banzhaf index is exactly as in 
	$\mathcal{G}$,
	i.e., 
	we have constructed a yes-instance of our control problem.

	\proofif
	Assume now that $(\phi, k, \ell)$ is a no-instance of \textsc{E-Exact-SAT$^{*}$},
	i.e., there does not exist any
	assignment of the variables $x_1,\ldots,x_k$ such that exactly $\ell$
	assignments of the remaining $n-k$ variables 
	satisfy the Boolean
	formula~$\phi$.  In other words, for each assignment of $x_1,\ldots,x_k$,
	there exist either less or more than $\ell$ assignments of $x_{k+1},\ldots,x_n$
	that satisfy~$\phi$.
	
	Let us start with the case considered in the previous implication, i.e., the case of deletion of $k$ players $A'$ that correspond to some truth assignment for $\phi$ (but this time, we consider any of the assignments): either
	\begin{align*}
		\PenroseBanzhaf(\mathcal{G}_{\setminus A'},1) &
		> \frac{2k(2^k -1)\ell + 2k\xi +  k\ell (2^{n+2}+2)}{2^{|N|-1-k}} \\
		& = \frac{2k 2^k \xi + k\ell (2^{n+2}+2^{k+1})2^{k}}{2^{|N|-1}}  = \PenroseBanzhaf(\mathcal{G},1), \textrm{ or }  \\
		\PenroseBanzhaf(\mathcal{G}_{\setminus A'},1) &
		< \frac{2k(2^k -1)\ell + 2k\xi + k\ell (2^{n+2}+2)}{2^{|N|-1-k}} \\
		& = \frac{2k 2^k \xi + k\ell (2^{n+2}+2^{k+1})2^{k}}{2^{|N|-1}}   = \PenroseBanzhaf(\mathcal{G},1).
	\end{align*}
	Thus we have constructed a no-instance of our control problem in the case we currently consider.

	Now, we are going to analyze
        all other possible deletions.
        First note that the number of coalitions from Case~$4$ is larger than the number of coalitions from any other case:	
	\begin{eqnarray}
          \label{m:thelargest:3}
	  k\ell 2^{n+2} \left(2^{k}-1\right) & > & k\ell 2^{k+1}\left(2^{k}-1\right),
          \\
          \label{m:thelargest:1or2}
	  k\ell 2^{n+2} \left(2^{k}-1\right) & > &
	  2k\xi \left(2^k - 1\right)
          ~>~ 2k\xi, \text{ and }
	  \\
          \label{m:thelargest:5and6}
	  k\ell 2^{n+2} \left(2^{k}-1\right) &  \stackrel{k\ge 4}{>} &
	  k\ell \left(2^{n+2} + 2^{k+1}\right).
	\end{eqnarray}
      (\ref{m:thelargest:3}) compares the number of coalitions from Case~$4$ with those from from Case~$3$; (\ref{m:thelargest:1or2}) compares Case~$4$ with Cases~$1$ and~$2$; and (\ref{m:thelargest:5and6}) compares Case~$4$ with Cases~$5$ and~$6$.

	Moreover, note that if we delete some $\eta_i$ players from sets $L_i$, $0\le \eta_i \le \delta_i$,
	$1\le i \le h$, we get in the indices of
	player~$1$ some value
	\begin{align*}
		\ell '  & = ~ 2^{\delta_1 - \eta_1} + \cdots + 2^{\delta_h - \eta_h}
	\end{align*}
	instead of $\ell$, but the denominator increases by a factor of~$2^{\eta_1+\cdots+\eta_h}$, so
	\begin{align*}
		2^{\eta_1+\cdots+\eta_h}\ell '  & = ~2^{\eta_1+\cdots+\eta_h} (2^{\delta_1 - \eta_1} + \cdots + 2^{\delta_h - \eta_h}) \\
		& = ~ 2^{\eta_2+\cdots+\eta_h} 2^{\delta_1} + \cdots +2^{\eta_1+\cdots+\eta_{h-1}} 2^{\delta_h} \\
		& \ge ~ \ell,
	\end{align*}
	i.e.,
	\begin{align}\label{ell}
		2^{\eta_1+\cdots+\eta_h}\ell '  & \ge ~ \ell.
	\end{align}
	Moreover, since $h\ge 2$, if we delete at least one from these players, the sign above changes to $>$. 
	We will use the notation of $\ell'$ for a new value after possible removal of the players from $L_i$ later in our analysis---each time, its value should be clear from the description of the case being analyzed.
	
	Now, let us consider the summands defined by Case~$6$. Let us consider removing $l_1 + l_2 + l_3 + l_4 = l \le k$ with $l_1 + l_2 + l_3 < k$ (note that if the latter is $k$, the index will increase by (\ref{m:thelargest:5and6}) and that the summands can decrease to zero only if we remove all $k$ players from $Z^{*}$) and with $l_4 = l_{4,1}+\cdots + l_{4,h}$  players from $V$, $Y^{**}$, $Z^{*}$, and $L_i$, respectively, i.e., players forming coalitions from Case~$6$. Let us analyze first when the numerator decreases the most so that we do not have to consider all possible combinations for fixed $l_1$, $l_2$, $l_3$, and $l_4$: 
	After removing the $l_2$ players from $Y^{**}$ (and since they have the same weight, they are symmetric and it does not matter which ones we delete), it will be impossible for player~$1$ to be pivotal for the coalitions formed, i.a., by the players from $V$ that need more than $n+2-l_2$ players from $Y^{**}$ to achieve their total weight of~$q-1$, so we do not need to delete these players from the game---it would not decrease player~$1$'s 
	index any more. Also, the players from $V$ forming at most number of coalitions after the deletion of the mentioned players from $Y^{**}$ are players with weight of form
	$$q-\left\lfloor\frac{n+2-l_2}{2}\right\rfloor y^{**}-z^{*}_{i'} - \left\lfloor\frac{\delta_i - l_{4,i}}{2}\right\rfloor\delta_i ',$$ 
	$1\le i' \le k$, 
	$1\le i \le h$, each forming 
	$${n+2-l_2 \choose \floor{\frac{n+2-l_2}{2}}}{\delta_i - l_{4,i} \choose \floor{\frac{\delta_i - l_{4,i}}{2}}}$$ coalitions with weight~$q-1$, so we get the largest decrease if we remove $l_1$ from those players with $i'$ such that the players with weight~$z^{*}_{i'}$ appear in the new game. Therefore, for 
	\[
	j = \max_{i\in\{1,\ldots,h\}}{\delta_i - l_{4,i} \choose \floor{\frac{\delta_i - l_{4,i}}{2}}},
	\]
	if $l_2\ge 1$, then the summand defined by Case~$6$ changes to at least 
	\begin{align*} 
	  & \frac{(k-l_3)(2^{n+2-l_2}+1)\ell ' - l_1 j {n+2-l_2 \choose \floor{\frac{n+2-l_2}{2}}}}{2^{|N|-1-l}} \\
          & =  \frac{2^l  \Big((k-l_3)(2^{n+2-l_2}+1)\ell' - l_1 j {n+2-l_2 \choose \floor{\frac{n+2-l_2}{2}}}\Big)}{2^{|N|-1}} \\
	  & \ge \frac{2^{l_4}\Big(2^{l_1 + l_3}(k-l_3)(2^{n+2}+2)\ell' - 2^{l-l_4}l_1 j {n+2-l_2 \choose \floor{\frac{n+2-l_2}{2}}}\Big)}{2^{|N|-1}} \\
		& = 2^{l_4}\frac{k(2^{n+2}+2)\ell' + (2^{l_1 + l_3}-1)k(2^{n+2}+2)\ell'}{2^{|N|-1}} \\
	  & ~~~~~ + 2^{l_4}\frac{ - 2^{l_1 + l_3}l_3 (2^{n+2}+2)\ell' - 2^{l-l_4}l_1 j {n+2-l_2 \choose \floor{\frac{n+2-l_2}{2}}}}{2^{|N|-1}}
	\end{align*}
	\begin{align*}
	  & \stackrel{(\ref{ell})}{\ge} \frac{k(2^{n+2}+2)\ell}{2^{|N|-1}}   \\ 
		& ~~~~~ + 2^{l_4}\frac{(2^{l_1 + l_3}-1)k(2^{n+2}+2)\ell'- 2^{l_1 + l_3}l_3 (2^{n+2}+2)\ell'}{2^{|N|-1}} \\
		& ~~~~~ - 2^{l_4}\frac{2^{l-l_4}l_1 j 2^{n+1-l_2}}{2^{|N|-1}} \\
		& =  \frac{k(2^{n+2}+2)\ell}{2^{|N|-1}} \\ 
		& ~~~~~ + 2^{l_4}\frac{ (2^{l_1 + l_3}-1)k(2^{n+2}+2)\ell' - 2^{l_1 + l_3}l_3 (2^{n+2}+2)\ell'}{2^{|N|-1}} \\
		& ~~~~~ - 2^{l_4}\frac{ l_1 j 2^{n+1+l_1 + l_3}}{2^{|N|-1}} \\
		& \ge \frac{k(2^{n+2}+2)\ell}{2^{|N|-1}} \\
		& ~~~~~ + 2^{l_4}(2^{n+2}+2)\ell'\frac{ (2^{l_1 + l_3}-1)k-  2^{l_1 + l_3 }(\frac{j}{2\ell'}l_1 + l_3)}{2^{|N|-1}}  
	\end{align*}
	and because $2^{l_1 + l_3}  \ge l_1 + l_3 +1$ for $l_1 + l_3 \ge 0$,  $k\ge l_1 + l_3 + 1$, and $\frac{j}{\ell'}\le 1$,
	\begin{align*}
		& (2^{l_1 + l_3}-1)k-  2^{l_1 + l_3 }(\frac{j}{2\ell'}l_1 + l_3) \\ & \ge (2^{l_1 + l_3}-1)(l_1+l_3+1)-  2^{l_1 + l_3 }(l_1 + l_3) \\
		&  = 2^{l_1 + l_3}-1 -l_1-l_3 \\
		& \ge 0,
	\end{align*}
	so 
	\begin{small}
		\begin{align}\label{case:6}
			\frac{(k-l_3)(2^{n+2-l_2}+1)\ell' - l_1 j {n+2-l_2 \choose \floor{\frac{n+2-l_2}{2}}}}{2^{|N|-1-l}} &  \ge \frac{k(2^{n+2}+2)\ell}{2^{|N|-1}}. 
		\end{align}
	\end{small}
	If $l_2=0$, then the value $2^{n+2-l_2}+1$ is replaced by $2^{n+2}+2$ and ${n+2-l_2 \choose \floor{\frac{n+2-l_2}{2}}}$ by ${n+2 \choose \floor{\frac{n+2}{2}}}$, and the calculations are analogous,
	i.e., the summand does not decrease (and if we delete also some $l'$ players from 
	$N\setminus (\{1\}\cup V \cup Y^{**} \cup Z^{*}\cup L)$, the value of the summand will increase by a factor of~$2^{l'}$).

	Now, let us consider only Case~$3$.
	Let $l_1,l_2,l_3,l_4$ with $l_1+l_2+l_3 +l_4 = l$ and $l_4=l_{4,1}+\cdots+l_{4,h}$ be the number of players being removed from~$S$, $Y$, $Z$, and~$L$, respectively, let $l'<k$ (the case of $l'=k$ was already analyzed) be the number of players deleted from $A$ such that no whole pair $(a_i,b_i)$ for any $i\in\{1,\ldots,k\}$ is deleted, and let $1\le l+l' \le k$. Let us consider the case giving us 
	the smallest
	number of coalitions for which $1$ stays pivotal. The players from $Y$ and $Z$ are symmetric, therefore it does not matter which ones we remove. In the new game, there are $k+1-l_2$ players with weight~$y$, 
	$k-l_3$ players with weight~$z$, and $k-l'$ whole pairs in $A$. Since player~$1$ cannot be pivotal for the coalitions that need more players from $Y$ or $Z$ than there are in the new game anymore, and for coalitions that need the removed players from $A$, we do not need to remove the players from $S$ forming coalitions that either need more players from $Y$ or $Z$ or the whole pair from $A$ that is not there anymore to achieve the total weight~$q-1$. The most number of coalitions (after deleting the $l_2 + l_3 + l'$ players) are formed by players from $S$ having weight of form
	\begin{small}	
		\[
		q-(a_{i'} + b_{i'}) - \left\lfloor\frac{k+1-l_2}{2}\right\rfloor y - \left\lfloor\frac{k+1-l_3}{2}\right\rfloor z- \left\lfloor\frac{\delta_{i} - l_{4,i}}{2}\right\rfloor\delta_{i} '
		\]
	\end{small}
	(each forming 
	$${k+1-l_2 \choose \floor{\frac{k+1-l_2}{2}}}
	{k-l_3 \choose \floor{\frac{k-l_3}{2}}}{\delta_{i} - l_{4,i} \choose \floor{\frac{\delta_{i} - l_{4,i}}{2}}}$$ 
	coalitions of weight $q-1$), and there are not less than $k-l'\ge l_1$ players with such a weight. Let us remove any $l_1$ players with such a weight. Let  
	$$j = \max_{i\in\{1,\ldots,h\}}{\delta_i - l_{4,i} \choose \floor{\frac{\delta_i - l_{4,i}}{2}}}.$$ 
	From Case~$3$ in the new game, we get the summand of at least the following value in player~$1$'s power index---we distinguish the following four cases: 
	\begin{enumerate}
		\item $l_3 > 0 \wedge l_1 = 0$ :
		\begin{small}
			\begin{align*}
				\frac{(k-l')2^{k+1-l_2}2^{k-l_3}\ell'}{2^{|N|-1-l-l'}}  = ~ & \frac{2^{l+l'}(k-l')2^{k+1-l_2}2^{k-l_3}\ell'}{2^{|N|-1}} \\
				> ~ & \frac{2^{l_4+l'}(k-l')2^{k+1}(2^{k}-1)\ell'}{2^{|N|-1}}  
			\end{align*}
		\end{small}
		and since $2^{l'}\ge l'+1$ for $l'\ge 0$, and $l'<k$,
		\begin{align*}
			(2^{l'}-1)k & \ge (2^{l'}-1)(l'+1) \\
			& = 2^{l'}l' + 2^{l'} - l' - 1 \\
			& \ge 2^{l'} l', 
		\end{align*}
		which implies
		\begin{align}\label{m:k:inequality}
			2^{l'}(k-l') & \ge k, 
		\end{align} 
		therefore,
		\begin{align*}
			\frac{(k-l')2^{k+1-l_2}2^{k-l_3}\ell'}{2^{|N|-1-l-l'}} \stackrel{(\ref{m:k:inequality})}{>} ~ & 
			2^{l_4}\frac{k2^{k+1}(2^{k}-1)\ell'}{2^{|N|-1}}  \\
			\stackrel{(\ref{ell})}{\ge} ~ & \frac{k2^{k+1}(2^{k}-1)\ell}{2^{|N|-1}} 
		\end{align*}
		
		\item $l_3 > 0 \wedge l_1 > 0$:
		\begin{small}
			\begin{align*}
				& \frac{ (k-l')2^{k+1-l_2}2^{k-l_3}\ell' - l_1 j {k+1-l_2 \choose \floor{\frac{k+1-l_2}{2}}}{k-l_3 \choose \floor{\frac{k-l_3}{2}}}}{2^{|N|-1-l-l'}} \\
				\ge ~ & \frac{(k-l')2^{k+1-l_2}2^{k-l_3}\ell' 
				- l_1 j 2^{k-l_2}2^{k-1-l_3}}{2^{|N|-1-l-l'}}%
				\\
				= ~ & \frac{2^{l_4}}{2^{|N|-1}}\Big(2^{l_1 + l'}(k-l')2^{k+1}2^{k}\ell' - l_1 j 2^{2k+l_1 +l'-1}\Big) \\
				= ~ & \frac{2^{l_4}\ell'}{2^{|N|-1}}\Big(2^{l_1 + l'}k2^{k+1}(2^{k}- 1)+ 2^{l_1 + l'}k2^{k+1}  \\
				&  - 2^{l_1 + l'}l'2^{k+1}2^{k} - l_1 \frac{j}{\ell'} 2^{2k+l_1 +l'-1}\Big) \\
				= ~ & 2^{l_4}\ell'\frac{k2^{k+1}(2^{k}-1)}{2^{|N|-1}} \\
				& + \frac{2^{l_4}\ell'}{2^{|N|-1}}
				\Big(k2^{2k+1+l_1+l'} - k2^{k+1+l_1+l'} - k2^{2k+1} \\
				&  + k2^{k+1} 
				+ k2^{k+1+l_1 +l'} 
				- l'2^{2k+1 + l_1 + l'}  \\
				& - l_1 \frac{j}{\ell'} 2^{2k+l_1 +l'-1}\Big) 
			\\
				= ~ & 2^{l_4}\ell'\frac{k2^{k+1}(2^{k}-1)}{2^{|N|-1}} \\
				& + \frac{2^{l_4}\ell'}{2^{|N|-1}}\Big((k-l')2^{2k+1 + l_1 + l'}   -k2^{2k+1} 
				+ k2^{k+1} \\
                                \\
	\end{align*}
	\begin{align*}
				\stackrel{(\ref{m:k:inequality})}{\ge} ~ & 2^{l_4}\ell'\frac{k2^{k+1}(2^{k}-1)}{2^{|N|-1}} \\
				& + \frac{2^{l_4}\ell'}{2^{|N|-1}}\Big((k-l')2^{2k + l_1 + l'} +k2^{2k + l_1} -k2^{2k+1}  \\
				& +  k2^{k+1}   - l_1 \frac{j}{\ell'} 2^{2k+l_1 +l'-1}\Big) \\
				\ge ~ & 2^{l_4}\ell'\frac{k2^{k+1}(2^{k}-1)}{2^{|N|-1}} \\
				& + \frac{2^{l_4}\ell'}{2^{|N|-1}}\Big((k-l')2^{2k + l_1 + l'-1} 
				+ k2^{k+1} \Big) \\
				\stackrel{(\ref{ell})}{>}  ~ & \frac{k2^{k+1}(2^{k}-1)\ell}{2^{|N|-1}}.
			\end{align*}
		\end{small}
		
		\item $l_3 = 0 \wedge l_1 = 0$:
		\begin{align*} 
			& \frac{(k-l')2^{k+1-l_2}(2^{k}-1)\ell'}{2^{|N|-1-l_2-l_4-l'}} \\
			= ~ & 2^{l_4}\ell'\frac{2^{l_2 +l'}(k-l')2^{k+1-l_2}(2^{k}-1)}{2^{|N|-1}} \\
			= ~ & 2^{l_4}\ell'\frac{2^{l'}(k-l')2^{k+1}(2^{k}-1)}{2^{|N|-1}} \\ 
			\stackrel{(\ref{ell}),(\ref{m:k:inequality})}{\ge} & ~~ \frac{k2^{k+1}(2^{k}-1)\ell}{2^{|N|-1}}.
		\end{align*}
		
		\item $l_3 = 0 \wedge l_1 > 0$:
		\begin{small}
			\begin{align*}
				& \frac{(k-l')2^{k+1-l_2}(2^{k}-1)\ell' - l_1 j {k+1-l_2 \choose \floor{\frac{k+1-l_2}{2}}}{k \choose \floor{\frac{k}{2}}}}{2^{|N|-1-l-l'}} \\
				\ge ~ & \frac{(k-l')2^{k+1-l_2}(2^{k}-1)\ell' 
				- l_1 j 2^{k-l_2}2^{k-1}}{2^{|N|-1-l-l'}}%
				\\
				= ~ & \frac{2^{l_4}\ell'}{2^{|N|-1}}\Big(2^{l_1 + l'}(k-l')2^{k+1}(2^{k}-1) 
				\\
				& - 2^{l_1 + l'} l_1 \frac{j}{\ell'} 2^{2k-1}\Big)  \\
				= ~ & 2^{l_4}\ell'\frac{k2^{k+1}(2^{k}-1)}{2^{|N|-1}} \\
				& + \frac{2^{l_4}\ell'}{2^{|N|-1}}\Big((2^{l_1 +l'}-1)k2^{k+1}(2^{k}-1) \\
				& -2^{l_1 + l'}l' 2^{k+1}(2^{k}-1)- l_1 \frac{j}{\ell'} 2^{2k+l_1 + l'-1}\Big) \\
				= ~ & 2^{l_4}\ell'\frac{k2^{k+1}(2^{k}-1)}{2^{|N|-1}} \\
				& + \frac{2^{l_4}\ell'}{2^{|N|-1}}\Big((2^{k+l_1 +l'}-2^{l_1 +l'}-2^{k}+1)k2^{k+1} \\
				& -l' 2^{2k+1+l_1 +l'} + l'2^{k+1+l_1 +l'}- l_1 \frac{j}{\ell'} 2^{2k+l_1 + l'-1}\Big) \\
				= ~ & 2^{l_4}\ell'\frac{k2^{k+1}(2^{k}-1)}{2^{|N|-1}} 
				+ \frac{2^{l_4}\ell'}{2^{|N|-1}}\Big(k2^{2k+1+l_1 +l'}\\
				& -k2^{k+1 +l_1 +l'}-k2^{2k+1}+k2^{k+1} 
				-l' 2^{2k+1+l_1 +l'}\\
				& + l'2^{k+1+l_1 +l'}- l_1 \frac{j}{\ell'} 2^{2k+l_1 + l'-1}\Big)  
	\end{align*}
	\begin{align*}
				= ~ & 2^{l_4}\ell'\frac{k2^{k+1}(2^{k}-1)}{2^{|N|-1}} \\
				& + \frac{2^{l_4}\ell'}{2^{|N|-1}}\Big((k-l')2^{2k+l_1 +l'}+(k-l')2^{2k+l_1 +l'-1} \\
				& +(k-l')2^{2k+l_1 +l'-1} - (k-l')2^{k+1 +l_1 +l'}-k2^{2k+1} \\
				& +k2^{k+1} - l_1 \frac{j}{\ell'} 2^{2k+l_1 + l'-1}\Big)  \\
				\ge ~ & 2^{l_4}\ell'\frac{k2^{k+1}(2^{k}-1)}{2^{|N|-1}} \\
				& + \frac{2^{l_4}\ell'}{2^{|N|-1}}\Big((k-l')2^{2k+l_1 +l'}+(k-l')2^{2k+l_1 +l'-1} \\
				& - (k-l')2^{k+1 +l_1 +l'}-k2^{2k+1}+k2^{k+1} \Big).  
			\end{align*}
		\end{small}
		And since
		\begin{align*}
			& (k-l') 2^{2k+l_1 +l'} +(k-l')2^{2k+l_1 +l'-1}  \\
			& - (k-l')2^{k+1 +l_1 +l'}-k2^{2k+1}+k2^{k+1}	\\
			\stackrel{(\ref{m:k:inequality})}{\ge} ~ & (k-l')2^{2k+l_1 +l'-1}- (k-l')2^{k+1 +l_1 +l'}+k2^{k+1} \\
			> ~ & k2^{k+1} \\
			> ~ & 0, \\
		\end{align*}
		the summand increases also in this case.
	\end{enumerate}
	
	Note that if we removed some whole pair $(a_i,b_i)$, $i\in\{1,\ldots,k\}$, from $A$, we would get exactly the same number of coalitions for the summands as in the case of removing only one of the players from the pair, i.e., the new game would just contain less players, therefore, the summand would be even larger (since the denominator would be smaller at that time).

	Let us analyze the summand by Case~$5$. Let us consider removing $l_1 + l_2 + l_3 + l' = l \le k$ with $l'<k$ (if we remove the whole pair 
	$(a_{i'},b_{i'})$, the summand will not decrease more than by removing just one of the elements, so we assume the latter; then, $l'=k$ was considered separately) and $l_1+l_2<k$ (if $l_1+l_2=k$, the index will increase by (\ref{m:thelargest:5and6})) be the numbers of removed players from $U$, $Y^{*}$, $L$, and $A$, 
	respectively. As in the previous case, we can remove any players from $Y^{*}$. Then the most coalitions are formed with the players with weights of form
	$$q-(a_{i'}+b_{i'})-\left\lceil\frac{k+1-l_2}{2}\right\rceil y^{*}- \left\lfloor\frac{\delta_i - l_{3,i}}{2}\right\rfloor\delta_i ',$$ 
	for $i'$ for which $a_{i'}$ and $b_{i'}$ were not removed. 
	Let
	\[
	j = \max_{i\in\{1,\ldots,h\}}{\delta_i - l_{3,i} \choose \floor{\frac{\delta_i - l_{3,i}}{2}}}.
	\]
	Let us start with $l_2=0$. Then the summand defined by Case~$5$ changes to at least
	\begin{align*}
		& \frac{(k-l')(2^{k+1}-2)\ell ' - l_1 j {k+1 \choose \ceil{\frac{k+1}{2}}}}{2^{|N|-1-l}}.
	\end{align*}	
	If $l_1=0$, then by (\ref{ell}), the summand nondecreases. Otherwise,
	\begin{small}
		\begin{align*}
			& \frac{(k-l')(2^{k+1}-2)\ell ' - l_1 j {k+1 \choose \ceil{\frac{k+1}{2}}}}{2^{|N|-1-l}} \\
			& \ge \frac{2^{l_1+l_3+l'}\ell'}{2^{|N|-1}}\Big((k-l')(2^{k+1}-2) - l_1 \frac{j}{\ell'} 2^{k}\Big) 
\\
			& = 2^{l_3}\ell'\frac{k(2^{k+1}-2)}{2^{|N|-1}} 
			+ \frac{2^{l_3}\ell'}{2^{|N|-1}}\Big((2^{l_1+l'}-1)k(2^{k+1}-2)\\
			& ~~~~~ -2^{l_1+l'} l' (2^{k+1}-2) - l_1 \frac{j}{\ell'} 2^{k+l_1+l'}\Big) \\
			& = 2^{l_3}\ell'\frac{k(2^{k+1}-2)}{2^{|N|-1}} 
			+ \frac{2^{l_3}\ell'}{2^{|N|-1}}\Big(k2^{k+l_1+l'+1}\\
			& ~~~~~ - k2^{k+1} - k2^{l_1+l'+1}+2k - l'2^{k+l_1+l'+1} \\
			& ~~~~~ + l'2^{l_1+l'+1} - l_1 \frac{j}{\ell'} 2^{k+l_1+l'}\Big). \\
		\end{align*}
	\end{small}
	Since $l_1>0$, the summand by Case~$3$ increases at least by 
	\begin{align*}
		\alpha & =(2^{l_1}-1)k\ell 2^{k+1}(2^{k}-1) \\
		& =  (2^{k+l_1}-2^{l_1}-2^{k}+1) k\ell 2^{k+1} \\
		& \ge (2^{k+l_1-1}-2^{l_1}+1) k\ell 2^{k+1} \\
		& \ge (2^{l_1+l_3+l_1-1}-2^{l_1}+1) k\ell' 2^{k+1} \\
		& > 2^{l_1+l_3+l_1-2} k\ell' 2^{k+1}, \\
	\end{align*}
	so
	\begin{align*}
		& \frac{(k-l')(2^{k+1}-2)\ell ' - l_1 j {k+1 \choose \ceil{\frac{k+1}{2}}}}{2^{|N|-1-l}} + \alpha\\
		& > 2^{l_3}\ell'\frac{k(2^{k+1}-2)}{2^{|N|-1}} 
		+ \frac{2^{l_3}\ell'}{2^{|N|-1}}\Big((k-l')2^{k+l_1+l'+1}\\
		& ~~~~~ - k2^{k+1} - (k-l')2^{l_1+l'+1}+2k  \\
		& ~~~~~ - l_1 \frac{j}{\ell'} 2^{k+l_1+l'} + k2^{k+2l_1-1} \Big) \\
		& \ge 2^{l_3}\ell'\frac{k(2^{k+1}-2)}{2^{|N|-1}} 
		+ \frac{2^{l_3}\ell'}{2^{|N|-1}}\Big((k-l')2^{k+l_1+l'+1}\\
		& ~~~~~ - (k-l')2^{l_1+l'+1}+2k  - l_1 \frac{j}{\ell'} 2^{k+l_1+l'}  \Big) \\
		& \ge 2^{l_3}\ell'\frac{k(2^{k+1}-2)}{2^{|N|-1}} \\
		& ~~~~~ + \frac{2^{l_3}\ell'}{2^{|N|-1}}\Big((k-l')2^{k+l_1+l'} +2k  - l_1 \frac{j}{\ell'} 2^{k+l_1+l'}  \Big) \\
	\end{align*}
	and because $(k-l') \ge l_1 \ge l_1 \frac{j}{\ell'}$ and by (\ref{ell}), the sum of the two summands (i.e., by Case~$5$ and Case~$3$) increases in the case of $l_2=0$.
	Let  $l_2>0$ and $l_1>0$. Then the summand defined by Case~$5$ changes to at least
	\begin{small}
		\begin{align*} 
			& \frac{(k-l')(2^{k+1-l_2}-1)\ell ' - l_1 j {k+1-l_2 \choose \ceil{\frac{k+1-l_2}{2}}}}{2^{|N|-1-l}} \\ 
			& =  \frac{2^l  \Big((k-l')(2^{k+1-l_2}-1)\ell' - l_1 j {k+1-l_2 \choose \ceil{\frac{k+1-l_2}{2}}}\Big)}{2^{|N|-1}} \\
			& = 2^{l_3}\frac{2^{l_1  + l'}(k-l')(2^{k+1}-2^{l_2})\ell' - 2^{l-l_3}l_1 j {k+1-l_2 \choose \ceil{\frac{k+1-l_2}{2}}}}{2^{|N|-1}} \\
			& = 2^{l_3}\frac{2^{l_1  + l'}(k-l')(2^{k+1}-2)\ell' - 2^{l_1  + l'}(k-l')(2^{l_2}-2)\ell'}{2^{|N|-1}} \\
			& ~~~~~ - 2^{l_3}\frac{ 2^{l-l_3}l_1 j {k+1-l_2 \choose \ceil{\frac{k+1-l_2}{2}}}}{2^{|N|-1}} 
	\\	
			& = 2^{l_3}\frac{k\ell' (2^{k+1}-2)}{2^{|N|-1}}\\
			& ~~~~~ + 2^{l_3}\frac{(2^{l_1 + l'}-1)k(2^{k+1}-2)\ell'-2^{l_1  + l'}l'(2^{k+1}-2)\ell' }{2^{|N|-1}} \\
			& ~~~~~ + 2^{l_3}\frac{ - 2^{l_1  + l'}(k-l')(2^{l_2}-2)\ell' - 2^{l-l_3}l_1 j {k+1-l_2 \choose \ceil{\frac{k+1-l_2}{2}}}}{2^{|N|-1}} \\		
			& \ge 2^{l_3}\frac{k\ell' (2^{k+1}-2)}{2^{|N|-1}}\\
			& ~~~~~ + 2^{l_3}\frac{(2^{l_1 + l'}-1)k(2^{k+1}-2)\ell'-2^{l_1  + l'}l'(2^{k+1}-2)\ell' }{2^{|N|-1}} \\
			& ~~~~~ + 2^{l_3}\frac{ - 2^{l_1  + l'}(k-l')(2^{l_2}-2)\ell' - 2^{l_1+l'}l_1 j 2^{k}}{2^{|N|-1}} \\
			& =2^{l_3}\frac{k\ell' (2^{k+1}-2)}{2^{|N|-1}}\\
			& ~~~~~ + \frac{2^{l_3}\ell'}{2^{|N|-1}}\Big(k2^{k+l_1+l'+1} - k2^{k+1} - k2^{l_1+l'+1} \\
			& ~~~~~ + 2k - l' 2^{k+ l_1  + l'+1} + l'2^{l_1+l'+1}- k2^{l_1+l_2+l'}\\
			& ~~~~~  + k2^{l_1+l'+1} + l' 2^{l_1+l_2+l'} - l'2^{l_1+l'+1}   - l_1 \frac{j}{\ell'}2^{k+l_1+l'}\Big) \\
			& =2^{l_3}\frac{k\ell' (2^{k+1}-2)}{2^{|N|-1}}\\
			& ~~~~~ + \frac{2^{l_3}\ell'}{2^{|N|-1}}\Big(k2^{k+l_1+l'+1} - k2^{k+1}  + 2k  \\
			& ~~~~~ - k2^{l_1+l_2+l'}  + l' 2^{l_1+l_2+l'} - l'2^{k+l_1+l'+1}  -l_1 \frac{j}{\ell'}2^{k+l_1+l'}\Big) \\
		\end{align*}
	\end{small}
	and for $\alpha$ as defined earlier,
	\begin{align*} 
		& \frac{(k-l')(2^{k+1-l_2}-1)\ell ' - l_1 j {k+1-l_2 \choose \ceil{\frac{k+1-l_2}{2}}}}{2^{|N|-1-l}} + \alpha \\ 
		& > 2^{l_3}\frac{k\ell' (2^{k+1}-2)}{2^{|N|-1}}\\
		& ~~~~~ + \frac{2^{l_3}\ell'}{2^{|N|-1}}\Big(k2^{k+l_1+l'+1} - k2^{k+1}  + 2k  - k2^{l_1+l_2+l'} \\
		&  ~~~~~ + l' 2^{l_1+l_2+l'}   - l'2^{k+l_1+l'+1}  -l_1 \frac{j}{\ell'}2^{k+l_1+l'} 
		+ k2^{k+2l_1-1}\Big) \\
			\end{align*}
		\begin{align*}
		& \ge 2^{l_3}\frac{k\ell' (2^{k+1}-2)}{2^{|N|-1}}\\
		& ~~~~~ + \frac{2^{l_3}\ell'}{2^{|N|-1}}\Big((k-l')2^{k+l_1+l'+1}  + 2k \\
		& ~~~~~  - (k-l')2^{l_1+l_2+l'}   -l_1 \frac{j}{\ell'}2^{k+l_1+l'} \Big) \\
		& \ge 2^{l_3}\frac{k\ell' (2^{k+1}-2)}{2^{|N|-1}} \\
		& ~~~~~ + \frac{2^{l_3}\ell'}{2^{|N|-1}}\Big((k-l')2^{k+l_1+l'}  + 2k    -l_1 \frac{j}{\ell'}2^{k+l_1+l'} \Big). 
	\end{align*}
	As before, the sum of the two summand increases if $l_2>0$ and $l_1>0$. Finally, let $l_2>0$ and $l_1=0$. Then
	the summand defined by Case~$5$ changes to at least	
	\begin{small}
		\begin{align*} 
			& \frac{(k-l')(2^{k+1-l_2}-1)\ell ' }{2^{|N|-1-l}} \\ 
			& = \frac{2^{l_2+l_3+l'}\ell'}{2^{|N|-1}}(k-l')(2^{k+1-l_2}-1) 
	\\	
			& = \frac{2^{l_3+l'}\ell'}{2^{|N|-1}}(k-l')(2^{k+1}-2^{l_2}) \\
			& = \frac{2^{l_3+l'}\ell'}{2^{|N|-1}}\Big((k-l')(2^{k+1}-2)-(k-l')(2^{l_2}-2)\Big) \\
			& = 2^{l_3}\frac{k\ell' (2^{k+1}-2)}{2^{|N|-1}} \\
			& ~~~~~ + \frac{2^{l_3}\ell'}{2^{|N|-1}}\Big((2^{l'}-1)k(2^{k+1}-2) - 2^{l'}l'(2^{k+1}-2) \\
			&  ~~~~~ -2^{l'}(k-l')(2^{l_2}-2)\Big) \\
			& = 2^{l_3}\frac{k\ell' (2^{k+1}-2)}{2^{|N|-1}} \\
			& ~~~~~ + \frac{2^{l_3}\ell'}{2^{|N|-1}}\Big(k2^{k+l'+1} - k2^{k+1} - k2^{l'+1}+ 2k - l'2^{k+l'+1} \\
			& ~~~~~  + l'2^{l'+1} - k2^{l_2+l'} + k2^{l'+1} + l'2^{l_2+l'} - l'2^{l'+1} \Big) \\
			& = 2^{l_3}\frac{k\ell' (2^{k+1}-2)}{2^{|N|-1}} \\
			& ~~~~~ + \frac{2^{l_3}\ell'}{2^{|N|-1}}\Big(k2^{k+l'+1} - k2^{k+1} + 2k - l'2^{k+l'+1} \\
			& ~~~~~ - k2^{l_2+l'}  + l'2^{l_2+l'}  \Big). \\
		\end{align*}
	\end{small}
	The summand by Case~$3$ increases at least by 
	\begin{align*}
		\gamma & = (2^{l_2}-1)k\ell 2^{k+1}(2^{k}-1) \\
		& =  (2^{k+l_2}-2^{l_2}-2^{k}+1) k\ell 2^{k+1} \\
		& \ge (2^{k+l_2-1}-2^{l_2}+1) k\ell 2^{k+1} \\
		& \ge (2^{l_2+l_3+l_2-1}-2^{l_2}+1) k\ell' 2^{k+1} \\
		& > 2^{2l_2+l_3-2} k\ell' 2^{k+1}. \\
	\end{align*}
	Therefore,
	\begin{align*} 	
		& \frac{(k-l')(2^{k+1-l_2}-1)\ell ' }{2^{|N|-1-l}} + \gamma\\
		& > 2^{l_3}\frac{k\ell' (2^{k+1}-2)}{2^{|N|-1}} \\
		& ~~~~~ + \frac{2^{l_3}\ell'}{2^{|N|-1}}\Big(k2^{k+l'+1} - k2^{k+1} + 2k - l'2^{k+l'+1}  - k2^{l_2+l'}  \\
		& ~~~~~ + l'2^{l_2+l'} + k2^{k+2l_2-1}  \Big) \\
		& \ge 2^{l_3}\frac{k\ell' (2^{k+1}-2)}{2^{|N|-1}} \\
		& ~~~~~ + \frac{2^{l_3}\ell'}{2^{|N|-1}}\Big((k-l')2^{k+l'+1}  + 2k   - (k-l')2^{l_2+l'}  \Big) \\
		& > 2^{l_3}\frac{k\ell' (2^{k+1}-2)}{2^{|N|-1}}, 
	\end{align*}
	i.e., the sum of the two summands are greater 
	than 
	the corresponding sum in the old game.
	\OMIT{
		and as before
		\[
		(2^{l_1 + l'}-1)k-  2^{l_1  + l' }(\frac{j}{2\ell'}l_1 + l') \ge 0,
		\]
		so
		\begin{align*}%
			\frac{(k-l')(2^{k+1-l_2}-1)\ell' - l_1 j {k+1-l_2 \choose \ceil{\frac{n+2-l_2}{2}}}}{2^{|N|-1-l}}
			&  > 2^{l_3}\frac{k(2^{k+1}-2)\ell'}{2^{|N|-1}} \\ 
			& \stackrel{(\ref{ell})}{\ge}  \frac{k(2^{k+1}-2)\ell}{2^{|N|-1}}. 
		\end{align*}
		If $l_2=0$, the inequality sign above changes to~$\ge$. 
	}

	Let us consider removing $l_1 + l_2 + l_3 + l_4 + l' = l + l' \le k$ players from $T$, $Y'$, $Z'$, $L$, and $A$, respectively, with the assumptions that $l'<k$ and no whole pair $(a_i,b_i)$ for any $i\in\{1,\ldots,k\}$ is removed from the game.
	Let
	\[
	j = \max_{i\in\{1,\ldots,h\}}{\delta_i - l_{4,i} \choose \floor{\frac{\delta_i - l_{4,i}}{2}}}.
	\]
	Before we analyze the summand by Case~$4$, let us prove that for any $k\ge 4$ and $l'\in\{1,\ldots,k-1\}$, the following inequality is true:
	\begin{align}\label{3k/2}
		2^{l'}(k-l')\ge\frac{3}{2}k.
	\end{align}
	We are going to show it by induction. First let $k=4$. Then
	\begin{itemize}
		\item for $l'=1$, $2^{l'}(k-l') - \frac{3}{2}k = 2\cdot (4-1) - \frac{3}{2}\cdot4 = 0$,
		\item for $l'=2$, $2^{l'}(k-l') - \frac{3}{2}k = 4\cdot (4-2) - 6 = 2$, and
		\item for $l'=3$, $2^{l'}(k-l') - \frac{3}{2}k = 8\cdot (4-3) - 6 = 2.$
	\end{itemize}
	Let us assume that the inequality~(\ref{3k/2}) is true for $k\ge 4$ and $l'\in\{1,\ldots,k-1\}$. Then, for $k'=k+1$,
	\[
	2^{l'}(k'-l') = 2^{l'}(k-l') + 2^{l'} \ge \frac{3}{2}k + 2^{l'} > \frac{3}{2}k + \frac{3}{2} = \frac{3}{2}k',
	\]
	and for $l'=k$,
	\begin{align*}
		2^{k}(k'-k) - \frac{3}{2}k'  & = 2^k - \frac{3}{2}k - \frac{3}{2} \ge k^2 - \frac{3}{2}k - \frac{3}{2} \\
		& = (k - \frac{3-\sqrt{33}}{4})(k - \frac{3+\sqrt{33}}{4}) \\
		& > 0.
	\end{align*}
	Let us focus on the summand defined by Case~$4$ and let $l'>0$ (note that if $l'=0$, we do not need to use (\ref{3k/2}), and the values are still not smaller than the summand in the old game; also then, the changes made for the summand by Case~$4$ keep the summand by Case~$1$ unchanged): 
	\begin{enumerate}
		\item For $l_3 > 0 \wedge l_1 = 0$, the summand by Case~$4$ to equal to
		\begin{align*}
			& \frac{(k-l')2^{n+2-l_2}2^{k-l_3}\ell'}{2^{|N|-1-l-l'}} 	\\
			= ~ &  2^{l_4}\ell'\frac{2^{l'}(k-l')2^{n+2}2^{k}}{2^{|N|-1}} \\
			> ~ &  2^{l_4}\ell'\frac{2^{l'}(k-l')2^{n+2}(2^{k}-1)}{2^{|N|-1}} \\
			\stackrel{(\ref{3k/2})}{\ge} ~ &  2^{l_4}\ell'\frac{k2^{n+2}(2^{k}-1)}{2^{|N|-1}} + 2^{l_4}\ell'\frac{\frac{k}{2}2^{n+2}(2^{k}-1)}{2^{|N|-1}}\\
			\ge ~ &  \frac{k2^{n+2}(2^{k}-1)\ell}{2^{|N|-1}} + 2^{l_4}\ell'\frac{2k2^{n}(2^{k}-1)}{2^{|N|-1}},
		\end{align*}
		and additionally,
		\[
		2^{l_4}\ell'\frac{2k2^{n}(2^{k}-1)}{2^{|N|-1}} \ge \frac{2k(2^{k}-1)\xi}{2^{|N|-1}}.
		\]
		So, even if the summand by Case~$1$ decreases to $0$, the two summands together increase.
		
		\item For $l_3>0 \wedge l_1 > 0$, the summand by Case~$4$ is at least
		\begin{small}	
			\begin{align*}
				& \frac{(k-l')2^{n+2-l_2}2^{k-l_3}\ell'- l_1 j {n+2-l_2 \choose \floor{\frac{n+2-l_2}{2}}}{k-l_3 \choose \floor{\frac{k-l_3}{2}}}}{2^{|N|-1-l-l'}}	\\
				\ge ~ & 2^{l_4}\ell'\frac{k2^{n+2} (2^{k}-1)}{2^{|N|-1}} 
				\\& + \frac{2^{l_4}\ell'}{2^{|N|-1}}\Big((2^{l_1+l'}-1)2^{n+2} k(2^{k}-1) + 2^{l_1+l'}k2^{n+2}  \\
				& -2^{l_1+l'}l'2^{n+2} 2^{k} -2^{l + l' - l_4}l_1 \frac{j}{\ell'} 2^{n+1 - l_2} 2^{k-l_3-1}\Big) \\
				= ~ & 2^{l_4}\ell'\frac{k2^{n+2} (2^{k}-1)}{2^{|N|-1}} 
				+ \frac{2^{n+2+l_4}\ell'}{2^{|N|-1}}\Big(k2^{k+l_1+l'}
				\\& -k2^{l_1 + l'}-k2^{k}+k+ k2^{l_1+l'} \\
				&  -l'2^{k+l_1+l'} -l_1 \frac{j}{\ell'} 2^{k+l_1 +l' -2}\Big) \\
				= ~ & 2^{l_4}\ell'\frac{k2^{n+2} (2^{k}-1)}{2^{|N|-1}} 
				\\& + \frac{2^{n+2+l_4}\ell'}{2^{|N|-1}}\Big(4(k-l')2^{k+l_1+l'-2}
				\\& -k2^{k}+k 
				-l_1 \frac{j}{\ell'} 2^{k+l_1 +l' -2}\Big) \\
				\ge ~ & 2^{l_4}\ell'\frac{k2^{n+2} (2^{k}-1)}{2^{|N|-1}} 
				\\& + \frac{2^{n+2+l_4}\ell'}{2^{|N|-1}}\Big(3(k-l')2^{k+l_1+l'-2}-k2^{k}+k \Big) \\
				\stackrel{(\ref{3k/2})}{>} ~ &  \frac{k2^{n+2} (2^{k}-1)\ell}{2^{|N|-1}} + 
				\frac{k2^{n+2}2^{k}\ell}{2^{|N|-1}}. \\
			\end{align*}
		\end{small}
		Then
		\[
		\frac{k2^{n+2}2^{k}\ell}{2^{|N|-1}} > \frac{2k2^{n}2^{k}}{2^{|N|-1}}\ge \frac{2k2^{k}\xi}{2^{|N|-1}},
		\]
		so the summand with the two summands by Case~$1$ and Case~$2$ increase even if the two latter decrease to $0$.	
		
		\item Let $l_3=0 \wedge l_1 = 0$.  
		The summand of Case~$4$ is equal to
		\begin{align*}
			& \frac{(k-l')2^{n+2-l_2}(2^{k}-1)\ell'}{2^{|N|-1-l_2-l_4-l'}}\\
			= ~ & 2^{l_4}\ell'\frac{2^{l'}(k-l')2^{n+2}(2^{k}-1)}{2^{|N|-1}} \\
			= ~ &  2^{l_4}\ell'\frac{k 2^{n+2} (2^{k}-1)}{2^{|N|-1}} \\
			& + 2^{l_4}\ell'\frac{(2^{l'}-1)k2^{n+2}(2^{k}-1)-2^{l'}l'2^{n+2}(2^{k}-1)}{2^{|N|-1}} \\
			= ~ &  2^{l_4}\ell'\frac{k 2^{n+2} (2^{k}-1)}{2^{|N|-1}}   \\
			&+ 2^{l_4}\ell'\frac{2^{l'}(k-l')2^{n+2}(2^{k}-1)-k2^{n+2}(2^{k}-1)}{2^{|N|-1}} \\
			\stackrel{(\ref{3k/2})}{\ge} ~ & \frac{k 2^{n+2} (2^{k}-1)\ell}{2^{|N|-1}} + 2^{l_4}\ell'\frac{\frac{k}{2}2^{n+2}(2^{k}-1)}{2^{|N|-1}} \\ 
		\end{align*}
		and
		\[
		\frac{\frac{k}{2}2^{n+2}(2^{k}-1)}{2^{|N|-1}} = \frac{2k2^{n}(2^{k}-1)}{2^{|N|-1}} \ge \frac{2k(2^{k}-1)\xi}{2^{|N|-1}}.
		\]

		\item For $l_3=0 \wedge l_1 > 0$, 
		the summand of Case~$1$ is as in the first two situations. The summand defined by Case~$4$ in the new game is equal to 
		\begin{small}
			\begin{align*}
				& \frac{(k-l')2^{n+2-l_2}(2^{k}-1)-l_1 j {n+2-l_2 \choose \floor{\frac{n+2-l_2}{2}}}{k \choose \floor{\frac{k}{2}}}}{2^{|N|-1-l-l'}}	\\
				\ge ~ & 2^{l_4}\ell'\frac{k2^{n+2} (2^{k}-1)}{2^{|N|-1}} \\
				& + \frac{2^{l_4}\ell'}{2^{|N|-1}}\Big((2^{l_1+l'}-1)2^{n+2} k(2^{k}-1)  \\
				&  -2^{l_1+l'}l'2^{n+2} (2^{k}-1) -2^{l_1 + l'}l_1\frac{j}{\ell'} 2^{n+k}\Big) \\
				= ~ &  2^{l_4}\ell'\frac{k2^{n+2} (2^{k}-1)}{2^{|N|-1}} 
				\\& + \frac{2^{n+2+l_4}\ell'}{2^{|N|-1}}\Big(k2^{k+l_1+l'}-k2^{l_1 + l'}-k2^{k}+k \\
				& -l'2^{k+l_1+l'} +l'2^{l_1 + l' } -l_1\frac{j}{\ell'} 2^{k+l_1 +l' -2}\Big) \\
				= ~ &  2^{l_4}\ell' \frac{k2^{n+2} (2^{k}-1)}{2^{|N|-1}} \\
				& + \frac{ 2^{n+2+l_4}\ell'}{2^{|N|-1}}\Big(4(k-l')2^{k+l_1+l'-2}-(k-l')2^{l_1 + l'} \\
				& -k2^{k}+k -l_1 \frac{j}{\ell'} 2^{k+l_1 +l' -2}\Big) \\
				\ge ~ &   2^{l_4}\ell'\frac{k2^{n+2} (2^{k}-1)}{2^{|N|-1}} 
				+ \frac{ 2^{n+2+l_4}\ell'}{2^{|N|-1}}\Big(2^{l_1+l'}(3(k-l')2^{k-2} \\
				& -(k-l'))-k2^{k}+k \Big) 
		\\	
				= ~ &   2^{l_4}\ell'\frac{k2^{n+2} (2^{k}-1)}{2^{|N|-1}} \\ 
				& + \frac{ 2^{n+2+l_4}\ell'}{2^{|N|-1}}\Big(2^{l_1+l'}(k-l')(3\cdot2^{k-2}-1)-k2^{k}+k \Big) \\
		\end{align*}
		\begin{align*}		
				\stackrel{(\ref{m:k:inequality})}{\ge} ~ &   2^{l_4}\ell'\frac{k2^{n+2} (2^{k}-1)}{2^{|N|-1}} \\
				& + \frac{ 2^{n+2+l_4}\ell'}{2^{|N|-1}}\Big(k2^{l_1}(3\cdot2^{k-2}-1)-k2^{k}+k \Big) \\
				= ~ &   2^{l_4}\ell'\frac{k2^{n+2} (2^{k}-1)}{2^{|N|-1}} \\
				& + \frac{ 2^{n+2+l_4}\ell'}{2^{|N|-1}}\Big(k2^{k+l_1-1} + k2^{k-2+l_1}-k2^{l_1}-k2^{k}+k \Big) \\
				\ge ~ &   2^{l_4}\ell'\frac{k2^{n+2} (2^{k}-1)}{2^{|N|-1}} 
				+ \frac{ 2^{n+2+l_4}\ell'}{2^{|N|-1}}\Big(k2^{l_1}(2^{k-2}-1)+k \Big) \\
				\ge ~ &   2^{l_4}\ell'\frac{k2^{n+2} (2^{k}-1)}{2^{|N|-1}} 
				+ \frac{ 2^{n+2+l_4}\ell'}{2^{|N|-1}}\Big(k(2^{k-1}-2)+k \Big) \\
				> ~ &  \frac{k2^{n+2} (2^{k}-1)\ell}{2^{|N|-1}} + \frac{ 2k2^{n+1}(2^{k-1}-2)}{2^{|N|-1}}, 
			\end{align*}
		\end{small}
		so the summand increases. Also, for $l'>0$, together with part of the increase of the summand by Case~$6$:
		\begin{align*}
			\frac{ 2k2^{n+1}(2^{k-1}-2)}{2^{|N|-1}} & + (2^{l_1+l'}-1)\frac{ 2k(2^{n+1}+1)}{2^{|N|-1}}\\ 
			& > \frac{ 2k2^{n}(2^{k}+2)}{2^{|N|-1}} \\
			&  > \frac{ 2k(2^{k}-1)\xi}{2^{|N|-1}}.
		\end{align*}
	\end{enumerate}  
	
	Now, let us consider more closely the last two cases: Case~$1$ and Case~$2$. Both of them could be 
	equal to $0$ from the beginning; then, removing any player is an increase of the index due to all the cases analyzed above (the players forming the coalitions from Case~$6$ and the players forming coalitions from the rest of the cases are disjoint subsets, i.e., removing any player from one subset automatically increases the summands defined by the other subset while the removal of the former do not decrease the summands it defines). 
	If the summands are positive, removal of one player for each can decrease it to~$0$. However, as it was shown while analyzing the previous summand and 
	by (\ref{m:thelargest:1or2}), we see that the index still increases.
	
	Summarizing the whole analysis above:
	\begin{enumerate}
		\item Removing $k$ players as it was done in the previous implication---but for any set corresponding to a truth assignment for the first $k$ variables---the index in the new game is not equal to the index in the old game.
		\item Removing players forming the coalitions from Case~$3$ makes the summand not smaller if we remove no player from $S\cup Z$, and greater if at least one of the players is deleted.
		\item If we remove some of the players forming coalitions 
		from Case~$5$ and Case~$6$, there are two possibilities, i.e.,
		\begin{itemize}
			\item either we remove $k$ players only from~$Z^{*}$ excluding all the coalitions from the index in the extreme situation, but by (\ref{m:thelargest:5and6}), the new index is greater due to the coalitions from Case~$4$, or
			\item we remove fewer of the players defining these summands---deleting or not other players forming coalitions from the other cases---then, the summands by Case~$5$ 	
			(together with the increase made by the summand by Case~$3$) 
			and by Case~$6$ nondecrease. 
		\end{itemize}
		\item  If we remove some players forming the coalitions from Case~$4$, this summand increases, unless we do not delete any player from $T\cup Z'$---then, the summand nondecreases. 
		\item If we remove the players that are crucial for the summand by Case~$1$ and it decreases even to $0$ (or $\xi=0$ from the beginning), the summand increases due to the increase of the summand by Case~$4$, if we remove at least one player from $A$. Without removing any player from $A$, the summands still increase due to, e.g., (\ref{m:thelargest:1or2}).
		\item If we remove the players that are crucial for the summands by Case~$2$, it increases due to  (\ref{m:thelargest:1or2}).
	\end{enumerate}
	\noindent
	Note that there exists no group of players that is part of forming some coalitions counted in the summands by all the six cases unless $\xi = 0$---in this situation, each positive summand is defined, i.a., by players from $L$.
        But first, let us discuss shortly the situation when $\xi > 0$.
        In this case, we have to remove players from at least two groups to alter  the number of coalitions counted in all summands, and at least two groups will have impact on
        at most four summands.
        Therefore, for each of these groups, after removing some players, the summands they impacted are not smaller than they were in the old game; however, at the same time, the removal causes an increase of the remaining summands, which are not smaller than in the old game either, after removing some players defining them.
        Therefore, the index increases at the end.
        
        Now, let $\xi=0$.
        In this case, we do not have to worry about a decrease of the summands by Case~$1$ and Case~$2$.
        Let us focus on the remaining four summands.
        If we remove a player from the groups defining only the summands by Case~$1$ and Case~$2$, the index increases.
        Since the summand by Case~$4$ is larger than any other summand, removing a player not in $A\cup T\cup Y' \cup Z' \cup L$ will increase the index due to the increase of this summand, i.e., removal of each such player will create a value larger than each of the remaining three summands (\ref{m:thelargest:3})--(\ref{m:thelargest:5and6}) (or even larger than two such summands, as in the case of (\ref{m:thelargest:5and6})).
        So let us assume we do not remove these players either.
        From the remaining groups, only groups~$A$ and $L$ define more summands, so removing players from the other three groups does not make the summand by 	
	Case~$4$ smaller, but it increases the former summands.
        Furthermore, let us assume that we do not delete these players either.
        The players from $A$ do not form coalitions counted in the summand by Case~$6$, so if this summand is not smaller due to removal of other players, it also increases due to the deletion of players from $A$ (while the remaining summands are not smaller than before).
        Finally, consider deleting only players from $L$: Let us call them $L'$ and let $|L'|>0$.
        Then
	\begin{align*}
		\frac{2^k k(2^{n+2}+2^{k+1})\ell '}{2^{|N|-|L'|-1}}
		= ~~~~ & 2^{|L'|}\ell'\frac{2^k k(2^{n+2}+2^{k+1})}{2^{|N|-1}} \\
		\stackrel{(\ref{ell})}{>} ~~~ & \frac{2^k k\ell(2^{n+2}+2^{k+1})}{2^{|N|-1}} \\
		= ~~~~ & \PenroseBanzhaf(\mathcal{G},1).
	\end{align*}

        Summing up, in each possible case of deleting players, the Penrose--Banzhaf index of player~$1$ differs in the new game from that in the old game, so we have constructed a no-instance of our control problem, as desired.
\end{proofs}

\end{document}